\documentclass[11pt,a4paper]{amsart}
\usepackage{bbm}
\usepackage{stmaryrd}
\usepackage[pdftex]{hyperref,color,graphicx}
\newtheorem{theorem}{Theorem}[section]

\newtheorem{lemma}[theorem]{Lemma}

\newtheorem{proposition}[theorem]{Proposition}

\theoremstyle{definition}
\newtheorem{definition}[theorem]{Definition}

\newtheorem{fact}[theorem]{Fact}

\theoremstyle{remark}

\newtheorem{remark}[theorem]{Remark}
\newtheorem{example}[theorem]{Example}

\numberwithin{equation}{section}

\newcommand{\Leb}{\mathrm{Leb}}
\newcommand{\conv}{\mathrm{conv}}
\newcommand{\cl}{\mathrm{cl}}
\newcommand{\polar}{\mathrm{polar}}
\newcommand{\bipolar}{\mathrm{bipolar}}
\newcommand{\solid}{\mathrm{solid}}

\newcommand{\ud}{\,\mathrm{d}}
\newcommand{\e}{\mathrm{e}}

\DeclareMathOperator*{\esssup}{ess\,sup}

\begin{document}
%
%
%
%
\title[Perpetual consumption duality]{Duality for optimal consumption
  under no unbounded profit with bounded risk}

\author[Michael Monoyios]{Michael Monoyios} 

\address{Mathematical Institute \\ 
University of Oxford \\
Radcliffe Observatory Quarter\\
Woodstock Road\\
Oxford OX2 6GG\\
UK}

\email{monoyios@maths.ox.ac.uk}

\date{\today}

\thanks{The author would like to thank Anastasiya Tanana for helpful
comments, as well as two anonymous referees and an Associate Editor
and Editor for constructive comments that improved the paper.}

\begin{abstract}

We give a definitive treatment of duality for optimal consumption
over the infinite horizon, in a semimartingale incomplete market
satisfying no unbounded profit with bounded risk (NUPBR). Rather than
base the dual domain on (local) martingale deflators, we use a class
of supermartingale deflators such that deflated wealth plus cumulative
deflated consumption is a supermartingale for all admissible
consumption plans. This yields a strong duality, because the enlarged
dual domain of processes dominated by deflators is naturally closed,
without invoking its closure. In this way we automatically reach the
bipolar of the set of deflators. We complete this picture by proving
that the set of processes dominated by local martingale deflators is
dense in our dual domain, confirming that we have identified the
natural dual space. In addition to the optimal consumption and
deflator, we characterise the optimal wealth process. At the optimum,
deflated wealth is a supermartingale and a potential, while deflated
wealth plus cumulative deflated consumption is a uniformly integrable
martingale. This is the natural generalisation of the corresponding
feature in the terminal wealth problem, where deflated wealth at the
optimum is a uniformly integrable martingale. We use no constructions
involving equivalent local martingale measures. This is natural, given
that such measures typically do not exist over the infinite horizon
and that we are working under NUPBR, which does not require their
existence. The structure of the duality proof reveals an interesting
feature compared with the terminal wealth problem. There, the dual
domain is $L^{1}$-bounded, but here the primal domain has this
property, and hence many steps in the duality proof show a marked
reversal of roles for the primal and dual domains, compared with the
proofs of Kramkov and Schachermayer \cite{ks99,ks03}.

\end{abstract}

\maketitle

\tableofcontents

\section{Introduction}
\label{sec:intro}

This paper gives a definitive treatment of duality for the optimal
consumption and investment problem for an agent maximising cumulative
discounted utility from consumption over an infinite horizon. This
problem has a long history, first being solved in a constant
coefficient complete Brownian model by Merton \cite{merton69} using
dynamic programming methods. The same model was studied in great
detail, considering also issues such as non-negativity constraints on
consumption, and bankruptcy, by Karatzas et al \cite{kal86} using
similar methods. Duality methods for a finite horizon version of the
problem to maximise utility from consumption and terminal wealth, in a
complete It\^o process market, were developed by Karatzas, Lehoczky
and Shreve \cite{kls87}. The infinite horizon problem for utility from
consumption in a complete It\^o market was treated via duality methods
by Huang and Pag\`es \cite{hp92}. In an incomplete It\^o market,
duality methods for the finite horizon problem of maximising utility
from terminal wealth were developed in a seminal paper by Karatzas et
al \cite{klsx91}. These methods were extended to finite horizon
problems including consumption and portfolio constraints (including
market incompleteness) by Cvitani\'{c} and Karatzas \cite{ck92} and
Shreve and Xu \cite{sx92}. Duality methods in an incomplete market
with general semimartingale asset prices were then developed for the
terminal wealth problem in a masterly contribution by Kramkov and
Schachermayer \cite{ks99,ks03}.

The finite horizon consumption problem in a semimartingale market,
under the no-arbitrage condition of No Free Lunch with Vanishing Risk
(NFLVR), was given a dual treatment by Karatzas and \v{Z}itkovi\'{c}
\cite{kz03} (who also incorporated a random endowment), building on
earlier work by \v{Z}itkovi\'{c} \cite{z02}. The infinite horizon
consumption problem remained an open problem to treat via duality
methods until fairly recently, when a significant advance was made by
Mostovyi \cite{most15}. Working under NFLVR, Mostovyi \cite{most15}
was able to show that most of the tenets of duality theory for utility
maximisation, as espoused by Kramkov and Schachermayer
\cite{ks99,ks03}, do hold true for the infinite horizon consumption
problem. This was extended by Chau et al \cite{chauetal17} to cover
the case where the no-arbitrage condition was weakened to the No
Unbounded Profit with Bounded Risk (NUPBR) condition, so that one need
not insist on the existence of equivalent local martingale measures
(ELMMs). This is a general observation, first made in explicit terms
by Karatzas and Kardaras \cite{kk07}, that all one needs for a
well-posed utility maximisation problem is the existence of a suitable
class of dual variables, or deflators, which need not be densities of
ELMMs, and which multiplicatively deflate primal variables to create
local martingales or supermartingales. This fact was implicit in
Karatzas et al \cite{klsx91}, which did not use ELMMs at all, and to
some extent was an underlying theme in the work of Kramkov and
Schachermayer \cite{ks99,ks03} who, despite working under NFLVR (so
ELMMs were definitively assumed to exist), expanded the dual domain to
a class of supermartingale deflators and found counter-examples where
the dual minimiser was not the density of an ELMM. We note that in
both Mostovyi \cite{most15} and Chau et al \cite{chauetal17} the
formulation could encompass other problems, by varying the measure (a
stochastic clock) that was used to aggregate utility from consumption
over time. (These papers also incorporated the stochastic clock into
the wealth dynamics, which amounts to a change of variable from a
traditional consumption rate, and we shall say more on this below.) By
varying this clock the approach in \cite{most15,chauetal17} can treat
the finite horizon utility from consumption problem, the terminal
wealth problem, as well as the finite horizon problem of utility from
both consumption and terminal wealth.

Given the above history, it is as well to point out where there is
still work to do and, as this is the focus of this paper, let us now
turn to this and describe the contribution.

First, we obtain a stronger duality statement than in Mostovyi
\cite{most15} and Chau et al \cite{chauetal17}, in the following
sense. In \cite{most15} and \cite{chauetal17} the initial dual domain
was based either on martingale deflators (in \cite{most15}, working
under NFLVR) or on local martingale deflators (in \cite{chauetal17},
working under NUPBR). The dual domain was then defined as the closure
(in an appropriate topology) of processes dominated by some element of
the set of deflators in question. The authors of
\cite{most15,chauetal17} were forced into taking the aforementioned
closure in order to obtain a closed dual domain, which could then be
shown to be the bipolar of the original domain of deflators, and thus
also the polar of the primal domain. Contrast this with the result of
Kramkov and Schachermayer \cite[Lemma 4.1]{ks99} in the terminal
wealth problem. There, one begins with a dual domain of
supermartingales (such that deflated admissible wealth is a
supermartingale for all strategies), then enlarges this domain to
consider random variables dominated by the terminal value of some
deflator. No closure is taken, but it is nevertheless shown that the
enlarged dual domain is naturally closed, so one reaches the bipolar
of the set of deflators, and perfect bipolarity between the primal and
dual domains is achieved. Herein lies our first contribution: we are
able to extend the prescription of Kramkov and Schachermayer
\cite{ks99}. First, we base our dual domain on a set of
supermartingales, this time such that \emph{deflated wealth plus
  cumulative deflated consumption is a supermartingale} for all
admissible consumption plans. Then, again in the spirit of
\cite{ks99}, we enlarge the dual domain to encompass processes
dominated by the deflators. Crucially, no closure needs to be
taken. We show that the enlarged dual domain is closed in the
appropriate topology, so that we reach the bipolar of the original
domain of supermartingales and obtain the duality between the primal
and dual optimisation problems without having to take a closure in
defining the enlarged dual domain. Finally, we show that our enlarged
dual domain coincides with the closure of processes dominated by local
martingale deflators, that is, the dual domain used in Chau et al
\cite{chauetal17}. Thus, the set of processes dominated by local
martingale deflators is dense in our dual domain. This result
(Proposition \ref{prop:dense}) is confirmation that we have chosen the
dual domain in just the right way to achieve a strong duality
statement. The underlying bipolarity results are obtained by
exploiting the Stricker and Yan \cite{sy98} version of the Optional
Decomposition Theorem (ODT), which uses deflators rather than ELMMs,
so we do not use any constructions whatsoever involving equivalent
measures. We shall say more on this aspect very shortly.

The second strengthening of the results in Mostovyi \cite{most15} and
Chau et al \cite{chauetal17} is fundamental. In addition to the
optimal consumption, we characterise the associated optimal wealth
process (and by extension the optimal strategy). Somewhat
surprisingly, neither of \cite{most15} or \cite{chauetal17} (or the
earlier works \cite{z02,kz03}) made any statement whatsoever regarding
the optimal wealth. This turns out to be a satisfying analysis which
shows shows that, at the optimum, deflated wealth is a supermartingale
and also a potential, decaying to zero, while deflated wealth plus
cumulative deflated consumption at the optimum is a uniformly
integrable martingale. This is natural, though to the best of our
knowledge has not been shown before in a general semimartingale
infinite horizon consumption problem. It is the natural generalisation
of the Kramkov and Schachermayer \cite{ks99,ks03} terminal wealth
result that, at the optimum, deflated wealth is transformed from a
supermartingale to a uniformly integrable martingale.

The next aspect of our work concerns the use of, or more accurately
the avoidance of, any constructions involving ELMMs. We are working on
an infinite horizon, and it is well known that in this case hardly any
models will admit ELMMs, because the candidate change of measure
density is not a uniformly integrable martingale over the infinite
timescale. While this can be dealt with, by (for example) eliminating
the tail $\sigma$-algebra in some way when wishing to use equivalent
measures restricted to a finite horizon $\sigma$-field, we bypass any
such pitfalls by exploiting the Stricker and Yan \cite{sy98} version
of the ODT and so avoiding ELMMs. As we are working under NUPBR, where
ELMMs might not exist at all (a case is point is a stock driven by a
three-dimensional Bessel process, which we use in an example of a
utility maximisation problem in our framework in Section
\ref{sec:example}), it is natural to construct proofs which avoid any
use of ELMMs if possible, and this is what we do.

Finally, the proof of the main duality theorem in our approach reveals
an interesting structure of the consumption problem compared with the
terminal wealth problem. In contrast to \cite{most15,chauetal17}, we
do not incorporate a stochastic clock into the wealth dynamics, so our
consumption rate is with respect to calendar time. The change of
variable used in \cite{most15,chauetal17} was convenient in those
papers, as it allowed the authors to assume that a constant
``consumption'' stream was allowed. This amounts to, in essence, a
decaying real consumption rate. (It is manifestly the case that with a
true consumption rate, one cannot guarantee being able to consume at a
constant rate for ever.) By choosing to work with the real consumption
rate, two aspects of the problem's underlying structure emerge. First,
it naturally leads to the correct supermartingale constraint that one
should apply at the outset: that deflated wealth plus cumulative
deflated consumption is a supermartingale. This leads to the correct
choice of dual domain. Second, it reveals a role reversal for the
primal and dual domains compared with the terminal wealth problem of
Kramkov and Schachermayer \cite{ks99,ks03}. In \cite{ks99,ks03},
because the constant wealth $X^{0}\equiv 1$ is admissible, the dual
domain in bounded in $L^{1}(\mathbb{P})$. But in the consumption
problem it is the \emph{primal} domain that is bounded in $L^{1}$
(with respect to an appropriate measure).
This role reversal of the primal and dual domains then manifests
itself in the proofs. In numerous steps of the program, a method that
works for the primal domain in \cite{ks99,ks03} is diverted to the
dual domain here, and vice versa. A prime example is the proof of
conjugacy of the value functions: in the terminal wealth problem one
creates a compact subset of the primal domain so as to apply the
minimax theorem, and proves that the dual value function is the convex
conjugate of the primal value function. Here, instead, one creates the
compact subset in the dual domain, and applies a transformed minimax
theorem (replacing maximisation with minimisation, and a concave
function with a convex one, and so on) and proves that the primal
value function is the concave conjugate of the dual value
function. There are many other instances of this role reversal, which
will be pointed out in the course of the proof of the duality theorems
in Section \ref{sec:pdt}. In view of these facets, we choose to give a
complete and self-contained treatment of the duality proofs in their
entirety.

The rest of the paper is structured as follows. In Section
\ref{sec:market} we describe the financial market, the admissible
consumption plans, and the class of dual variables (consumption
deflators), alongside the alternatives such as local martingale
deflators. In Section \ref{sec:cpd} we formulate the primal and dual
problems. The main duality theorem (Theorem \ref{thm:consd}) is given
in Section \ref{sec:dt}. In Section \ref{sec:abpd} we give an abstract
version of the bipolarity relations (Proposition \ref{prop:abp})
between suitably defined primal and dual domains, an associated
abstract version of the duality theorem (Theorem \ref{thm:adt}), and
state Proposition \ref{prop:dense}, that the set of processes
dominated by local martingale deflators is dense in the set of
processes dominated by consumption deflators. The bipolarity relations
are proven in Section \ref{sec:bcbpr} by considering the infinite
horizon budget constraint for consumption, and showing that it is both
a necessary and sufficient condition for admissibility. Here, we
complete the discussion on ramifications of using an alternative
choice of dual domain based on local martingale deflators, and prove
Proposition \ref{prop:dense}. In Section \ref{sec:pdt} we prove the
abstract duality, then establish Proposition \ref{prop:owp}
characterising the optimal wealth process, followed by the concrete
duality theorem. In Section \ref{sec:example} we give an example with
power utility and a stock driven by a three-dimensional Bessel
process, with stochastic volatility and correlation, for which the
dual minimiser is a strict local martingale, fitting well into our
earlier program.

\section{The market}
\label{sec:market}

We have an infinite horizon financial market containing $d$ stocks and
a cash asset, on a complete stochastic basis
$(\Omega,\mathcal{F},\mathbb{F}:=(\mathcal{F}_{t})_{t\geq
  0},\mathbb{P})$, with the filtration $\mathbb{F}$ satisfying the
usual conditions of right-continuity and augmentation with the
$\mathbb{P}$-null sets of $\mathcal{F}$. We shall use the cash asset
as num\'eraire, so work with discounted quantities. The (discounted)
stock price vector is given by a positive $d$-dimensional c\`adl\`ag
semimartingale $S=(S^{1},\ldots,S^{d})$.

An agent with initial capital $x>0$ can trade the stocks and cash and
may consume at a non-negative c\`adl\`ag adapted rate
$c=(c_{t})_{t\geq 0}$, assumed to satisfy the minimal condition
$\int_{0}^{t}c_{s}\ud s<\infty$, almost surely, $\forall\,t\geq
0$. The associated wealth process is $X$, given by
\begin{equation}
X_{t} = x + (H\cdot S)_{t} - \int_{0}^{t}c_{s}\ud s, \quad t\geq
0, \quad x>0.
\label{eq:wealth}
\end{equation}
In \eqref{eq:wealth}, $(H\cdot S)$ denotes the stochastic integral and
the trading strategy $H$ is a predictable $S$-integrable vector
process for the number of units of each stock held. Write
\begin{equation*}
C_{t} :=\int_{0}^{t}c_{s}\ud s, \quad t\geq 0,
\end{equation*}
for the non-decreasing cumulative consumption process. Then, with
\begin{equation}
X^{0} := x + (H\cdot S)
\label{eq:sfw}  
\end{equation}
denoting the wealth process of a self-financing portfolio
corresponding to strategy $H$, we have the decomposition
\begin{equation}
X = X^{0} - C.
\label{eq:wdecomp}  
\end{equation}

\subsection{Admissible consumption plans}
\label{subsec:acp}

We will assume solvency at all times, so $X\geq 0$ almost surely in
\eqref{eq:wealth}. In this case, for a given $x>0$, we call the pair
$(H,c)$ (or $(X,c)$) an $x$-admissible investment-consumption
strategy. If, for a consumption process $c$ we can find a predictable
$S$-integrable process $H$ such that $(H,c)$ is an $x$-admissible
investment-consumption strategy, then we say that $c$ is an
$x$-admissible consumption process or, briefly, an admissible
consumption plan. Denote the set of $x$-admissible consumption plans
by $\mathcal{A}(x)$:
\begin{equation}
\mathcal{A}(x) := \left\{c\geq 0:\exists\, H \, \mbox{such
that} \, X:=x+(H\cdot S) - \int_{0}^{\cdot}c_{s}\ud s\geq
0,\,\mbox{a.s}\right\}, \quad x>0. 
\label{eq:Ax}
\end{equation}
For $x=1$ we write $\mathcal{A}\equiv\mathcal{A}(1)$, and we note that
$\mathcal{A}(x)=x\mathcal{A}$ for $x>0$. We observe that $\mathcal{A}$
is a convex set.

For $c\equiv 0$, the wealth process is that of a self-financing
portfolio, with wealth process $X^{0}$ as in \eqref{eq:sfw}. Define
$\mathcal{X}(x)$ as the set of almost surely non-negative
self-financing wealth processes with initial value $x>0$:
\begin{equation*}
\mathcal{X}(x) := \left\{X^{0}:X^{0}=x+(H\cdot S)\geq
  0,\,\mbox{a.s.}\right\}, \quad x>0.
\end{equation*}
As for the admissible consumption plans, we write
$\mathcal{X}\equiv\mathcal{X}(1)$, with $\mathcal{X}(x)=x\mathcal{X}$
for $x>0$, and we note that $\mathcal{X}$ is a convex set.

Given the wealth decomposition in \eqref{eq:wdecomp}, an equivalent
characterisation of the admissible consumption plans is that there
exists a self-financing wealth process which dominates cumulative
consumption (such a wealth process will necessarily be non-negative,
so will lie in $\mathcal{X}(x)$).

\subsection{Deflators for consumption plans}
\label{subsec:dfcp}

The dual domain for our infinite horizon utility maximisation problem
from inter-temporal consumption will be a specialisation of the one
used by Kramkov and Schachermayer \cite{ks99,ks03} for the terminal
wealth problem. We shall refer to the processes in the dual domain as
\emph{deflators} (or, sometimes, as \emph{consumption deflators}, if
we need to distinguish them from the corresponding deflators in the
absence of consumption).

Define the set of positive c\`adl\`ag processes such that deflated
wealth plus cumulative deflated consumption is a supermartingale for
every admissible consumption plan:
\begin{equation}
\mathcal{Y}(y) := \left\{Y>0,\,\mbox{c\`adl\`ag},\, Y_{0}=y:
\mbox{$XY+\int_{0}c_{s}Y_{s}\ud s$
is a supermartingale, $\forall\,c\in\mathcal{A}$}\right\}.
\label{eq:mcY}  
\end{equation}
Using $\mathcal{A}$ rather than $\mathcal{A}(x)$ in \eqref{eq:mcY} is
without loss of generality, given $\mathcal{A}(x)=x\mathcal{A},\,x>0$.
As usual, we write $\mathcal{Y}\equiv\mathcal{Y}(1)$ and we have
$\mathcal{Y}(y)=y\mathcal{Y}$ for $y>0$. In \eqref{eq:mcY}, the wealth
process $X$ is the one on the left-hand-side of \eqref{eq:wealth} or
\eqref{eq:wdecomp} with $x=1$, so incorporating consumption. We note
that, since $(X,c)\equiv(1,0)$ is an admissible consumption-investment
pair, each $Y\in\mathcal{Y}(y)$ is a supermartingale. The set
$\mathcal{Y}$ is easily seen to be convex.

In the case $c\equiv 0$ (which is admissible) we have that deflated
self-financing wealth is a supermartingale for any choice of
consumption deflator. Thus, the set $\mathcal{Y}(y)$ is included in
the set of \emph{wealth deflators} that were used by Kramkov and
Schachermayer \cite{ks99,ks03}. We shall write $Y^{0}$ to denote such
deflators, and the set of wealth deflators will be denoted by
$\mathcal{Y}^{0}(y)$:
\begin{equation*}
\mathcal{Y}^{0}(y) := \left\{Y^{0}>0,\,\mbox{c\`adl\`ag},\, Y^{0}_{0}=y:
\mbox{$X^{0}Y^{0}$ is a supermartingale, for all
$X^{0}\in\mathcal{X}$}\right\}.
\end{equation*}
As before, we write $\mathcal{Y}^{0}\equiv\mathcal{Y}^{0}(1)$ and we
have $\mathcal{Y}^{0}(y)=y\mathcal{Y}^{0}$ for $y>0$. Since
$X^{0}\equiv 1$ lies in $\mathcal{X}$, each
$Y^{0}\in\mathcal{Y}^{0}(y)$ is a supermartingale. The wealth
deflators are also known as \emph{supermartingale deflators}. Clearly,
the set $\mathcal{Y}^{0}$ is convex.

The set $\mathcal{Z}$ of \emph{local martingale deflators} (LMDs) is
composed of positive c\`adl\`ag local martingales $Z$ with unit
initial value such that deflated self-financing wealth $X^{0}Z$, for
all $X^{0}\in\mathcal{X}$, is a local martingale:
\begin{equation}
\mathcal{Z} := \left\{Z>0,\,\mbox{c\`adl\`ag},\, Z_{0}=1:
\mbox{$X^{0}Z$ is a local martingale, for all
$X^{0}\in\mathcal{X}$}\right\}. 
\label{eq:mcZ}
\end{equation}
Since the local martingale $X^{0}Z\geq 0$ for all
$X^{0}\in\mathcal{X}$, it is also a supermartingale and, since
$X^{0}\equiv 1$ lies in $\mathcal{X}$, each $Z\in\mathcal{Z}$ is also
a supermartingale. The set $\mathcal{Z}$ contains the density
processes of equivalent local martingale measures (ELMMs) in
situations where those would exist. We shall not, however, be using
any constructions involving ELMMs, even restricted to a finite
horizon. We shall say more on this in Section
\ref{subsubsec:completion}.

We observe that we have the inclusions
\begin{equation*}
\mathcal{Z} \subseteq \mathcal{Y} \subseteq \mathcal{Y}^{0}.  
\end{equation*}
(To see the first inclusion, recall the wealth decomposition in
\eqref{eq:wdecomp}. Applying the It\^o product rule to the process
$CZ$ gives
$XZ + \int_{0}^{\cdot}c_{s}Z_{s}\ud s = X^{0}Z -
\int_{0}^{\cdot}C_{s-}\ud Z_{s}$, the left-hand-side of which is
non-negative, with the right-hand-side a local martingale, so the
left-hand-side is a non-negative local martingale and thus a
supermartingale. Thus, any LMD $Z\in\mathcal{Z}$ also lies in
$\mathcal{Y}$.)

The standing no-arbitrage assumption we shall make is that the set of
supermartingale deflators is non-empty:
\begin{equation}
\mathcal{Y}^{0}\neq \emptyset.  
\label{eq:noarb}
\end{equation}
It is well-known that \eqref{eq:noarb} is equivalent to the no
unbounded profit with bounded risk (NUPBR) condition (also known as no
arbitrage of the first kind, or $\mathrm{NA}_{1}$), weaker than the no
free lunch with vanishing risk (NFLVR) condition, the latter requiring
the existence of ELMMs, which is often problematic over the infinite
horizon, as we discuss in Section \ref{subsubsec:completion}. There
are a number of equivalent characterisations of NUPBR, including that
the set $\mathcal{Z}$ of LMDs is non-empty: see Karatzas and Kardaras
\cite{kk07}, Kardaras \cite{k12}, Takaoka and Schweizer \cite{ts14}
and Chau et al \cite{chauetal17}, as well as the recent overview by
Kabanov, Kardaras and Song \cite{kks16}.

\subsubsection{Completion of the stochastic basis and equivalent
measures}
\label{subsubsec:completion}

As indicated earlier, we shall avoid completely any constructions
which invoke equivalent local martingale measures (ELMMs), even
restricted to a finite horizon. This is partly for aesthetic reasons:
since we work under NUPBR and assume only the existence of various
classes of deflators, which is the minimal requirement for well posed
utility maximisation problems, it seems natural to seek proofs which
use only deflators. This is what we do.

There is some mathematical rationale for avoiding ELMMs.  We are
working on an infinite horizon and have have assumed the usual
conditions. Thus, each element of the filtration
$\mathbb{F}=(\mathcal{F}_{t})_{t\geq 0}$ includes all the
$\mathbb{P}$-null sets of
$\mathcal{F}:=\sigma(\bigcup_{t\geq
  0}\mathcal{F}_{t})=:\mathcal{F}_{\infty}$, the tail
$\sigma$-algebra. So, ultimate events (as time $t\uparrow\infty$) of
measure zero are included in any finite time $\sigma$-field
$\mathcal{F}_{T},\,T<\infty$.

It is well-known that in such a scenario many financial models will
not admit an equivalent martingale measure over the infinite horizon,
because the candidate change of measure density is not a uniformly
integrable martingale. (This is true of the Black-Scholes model, see
Karatzas and Shreve \cite[Section 1.7]{ks98}.) One then has to proceed
with caution when invoking arguments which utilise equivalent
measures, by finding a consistent way to eliminate the tail
$\sigma$-algebra from the picture when restricting to a finite horizon
$T<\infty$.

One route forward is to not complete the space, as in Huang and
Pag\`es \cite{hp92}, in an infinite horizon consumption model in a
complete Brownian market. This is sound, though care is needed to
ensure that no results are used which require the usual hypotheses to
hold.

Another way to proceed, if one wishes to consider equivalent measures
restricted to a finite horizon $T<\infty$, is to augment the space
with null events of a $\sigma$-field generated over a finite horizon
at least as big as $T$, that is by
$\sigma\left(\bigcup_{0\leq t\leq T^{\prime}}\mathcal{F}_{t}\right)$,
for some $0\leq T\leq T^{\prime}<\infty$. This can be done in a
consistent way, and relies on an application of Carath\'eodory's
extension theorem (Rogers and Williams \cite[Theorem
II.5.1]{rwvol1}). One can then obtain equivalent measures in an
infinite horizon model when restricting such measures to any finite
horizon. This procedure is carried out in a Brownian filtration in
Karatzas and Shreve \cite[Section 1.7]{ks98}, with a cautionary
example \cite[Example 1.7.6]{ks98}, showing that augmenting the
$\sigma$-field generated by Brownian motion over any finite horizon
with null sets of the corresponding tail $\sigma$-algebra would render
invalid the construction of equivalent measures, even over a finite
horizon.

The message is that one has to be careful in using any constructions
involving equivalent measures, even restricted to a finite horizon,
when working in infinite horizon financial model.

We avoid having to invoke such fixes, since we avoid all constructions
involving ELMMs. In particular, in Section \ref{sec:bcbpr} we
establish bipolarity results between the primal and dual domains using
only the Stricker and Yan \cite{sy98} version of the optional
decomposition theorem, relying on deflators rather than equivalent
measures.


\section{The consumption problem and its dual}
\label{sec:cpd}

Let $U:\mathbb{R}_{+}\to\mathbb{R}$ be a utility function, strictly
concave, strictly increasing, continuously differentiable on
$\mathbb{R}_{+}$ and satisfying the Inada conditions
\begin{equation}
\lim_{x\downarrow 0}U^{\prime}(x) = +\infty, \quad
\lim_{x\to\infty}U^{\prime}(x) = 0.
\label{eq:inada}
\end{equation}
To guarantee a well-posed consumption problem, one could also impose
here the reasonable asymptotic elasticity condition of Kramkov and
Schachermayer \cite{ks99}:
\begin{equation}
\mathrm{AE}(U) :=
\limsup_{x\to\infty}\frac{xU^{\prime}(x)}{U(x)} < 1.   
\label{eq:AEU}
\end{equation}
The condition in \eqref{eq:AEU} was shown in \cite{ks99} to be a
minimal condition, in an arbitrary market model, to guarantee that the
terminal wealth utility maximisation problem satisfied all the tenets
of a general duality theory. It was later shown, again by Kramkov and
Schachermayer \cite{ks03}, that if one instead assumes a market model
such that the weak condition of a finite dual value function holds,
then this alternative set-up gives a consistent duality
theory. Furthermore, finiteness of the dual problem, along with a
minimal condition on the primal value function (to be finitely valued
for at least one value of initial capital) so as to exclude a trivial
problem, implies the reasonable asymptotic elasticity condition. For
this reason, we shall follow the spirit of \cite{ks03} and just impose
weak finiteness conditions on the primal and dual value functions so
as to exclude trivial problems, and then later make the (standard)
remark in the style of \cite[Note 2]{ks03} on how these are consistent
with \eqref{eq:AEU} (see Remark \ref{rem:aeu}).

Let $\kappa:(0,\infty)\to\mathbb R_{+}$ be a positive finite measure
which will determine how utility of consumption is discounted through
time, assumed to be almost surely absolutely continuous with respect
to Lebesgue measure and satisfying
\begin{equation}
\frac{\ud\kappa_{t}}{\ud t} \leq 1, \quad \mbox{almost surely},
\quad t\geq 0, \quad
\mathbb{E}\left[\int_{0}^{\infty}c_{t}\ud\kappa_{t}\right] \leq
K<\infty,\,\forall\, c\in\mathcal{A}, 
\label{eq:kappa}
\end{equation}
for some constant $K>0$. For later use, define the positive process
$\gamma=(\gamma_{t})_{t\geq 0}$ as the reciprocal of
$(\ud\kappa_{t}/\ud t)_{t\geq 0}$:
\begin{equation}
\gamma_{t} := \left(\frac{\ud\kappa_{t}}{\ud t}\right)^{-1}, \quad
t\geq 0.
\label{eq:gamma}
\end{equation}

Define the primal value function from optimal consumption by
\begin{equation}
u(x) := \sup_{c\,\in\mathcal{A}(x)}\mathbb{E}\left[
\int_{0}^{\infty}U(c_{t})\ud\kappa_{t}\right], \quad x>0.
\label{eq:vf}
\end{equation}
To exclude a trivial problem, we shall assume throughout that
$u(x)>-\infty$ for all $x>0$. This is guaranteed by the weak condition
that $\mathbb{E}\left[\int_{0}^{\infty}
\min[0,U(c_{t})]\ud\kappa_{t}\right]>-\infty$. 

The supremum in \eqref{eq:vf} is written as one over consumption
processes. This should not obscure the fact that an optimal
consumption process must also determine an associated optimal wealth
process (equivalently an optimal trading strategy). This is clear from
the definition in \eqref{eq:Ax}, where the consumption process is
defined with reference to the associated investment strategy. Indeed,
in traditional formulations of the problem, this is acknowledged in
the notation by writing the value function as a supremum over a pair
of controls involving either $(X,c)$ or $(H,c)$. Our goal is to find
an optimal consumption process $\widehat{c}$, but to also characterise
the associated optimal wealth process $\widehat{X}$. Note that no such
characterisation of the optimal wealth process was given in either of
Mostovyi \cite{most15} or Chau et al \cite{chauetal17}. This turns out
to be an interesting feature of the analysis, with a nice result
(Proposition \ref{prop:owp}) incorporated into the main duality
theorem: at the optimum, the deflated wealth process is a
supermartingale and a potential, while the deflated wealth plus
cumulative deflated consumption is a uniformly integrable
martingale. These results are the natural extensions of the result for
the terminal wealth problem in Kramkov and Schachermayer
\cite{ks99,ks03}, in which optimal deflated wealth is a uniformly
integrable martingale.

\begin{example}[Infinite horizon discounted utility from consumption]
\label{ex:ihdufc}

The example we are primarily interested in is the case where
$\ud\kappa_{t}=\e^{-\alpha t}\ud t$, for some positive impatience
parameter $\alpha>0$ (which could also be made stochastic).  In this
case we have $\gamma_{t}=\e^{\alpha t},\,t\geq 0$, which is the factor
which inflates the natural deflators in the dual problem, as we shall
see.

The problem in \eqref{eq:vf} is then
$\mathbb{E}\left[\int_{0}^{\infty}\e^{-\alpha t}U(c_{t})\ud
  t\right]\to\max!$ We shall illustrate the solution of such a
problem with a stock driven by a three-dimensional Bessel process, and
with stochastic volatility and correlation, in Example
\ref{examp:3dbessel}.

\end{example}

\subsection{On stochastic clocks}
\label{subsec:sc}

We discuss briefly some variations of the problem \eqref{eq:vf} which
can be incorporated into our framework (but which are not the main
focus of our analysis).

In Mostovyi \cite{most15} and Chau et al \cite{chauetal17} the measure
$\kappa$ is taken to be a \emph{stochastic clock}, that is, a
non-decreasing, c\`adl\`ag adapted process satisfying
\begin{equation*}
\kappa_{0}=0, \quad \kappa_{\infty} \leq K <\infty,\,\mbox{a.s.},
\quad \mathbb{P}[\kappa_{\infty}>0]>0,
\end{equation*}
for some finite positive constant $K$. As shown by Mostovyi
\cite[Examples 2.5--2.9]{most15}, by appropriate choice of the
stochastic clock a number of different problems can be included within
the framework of \eqref{eq:vf}, such as the terminal wealth problem,
the finite horizon consumption problem, the finite horizon consumption
and terminal wealth problem, as well as the infinite horizon problem
in Example \ref{ex:ihdufc}. The same observation applies to our
problem, provided we choose the measure $\kappa$ to be a stochastic
clock of the appropriate type. Our primary focus, however, is to give
a definitive treatment of the traditional infinite horizon discounted
utility of consumption problem.

Note also that in \cite{most15,chauetal17}, the stochastic clock was
incorporated into the wealth dynamics: for some process
$\bar{c}$, \eqref{eq:wealth} was replaced by
\begin{equation*}
X_{t} = x + (H\cdot S)_{t} - \int_{0}^{t}\bar{c}_{s}\ud\kappa_{s},
\quad t\geq 0, \quad x>0.
\end{equation*}
Thus, the process $\bar{c}$ (let us call it a pseudo-consumption rate,
to distinguish it from our variable) of those papers involves a change
of variable from our consumption rate. The approach in
\cite{most15,chauetal17} allows for a constant positive
pseudo-consumption rate, which can sometimes be mathematically
convenient. With a true consumption rate and an infinite horizon, a
constant consumption plan is not possible. Each approach can be
converted to the other, as we now illustrate.

For concreteness, suppose the measure $\kappa$ is as in
\eqref{eq:kappa}. The pseudo-consumption rate $\bar{c}$ is then
related to the real consumption rate by
$\bar{c}_{t}=\gamma_{t}c_{t},\,t\geq 0$. The problems considered in
\cite{most15,chauetal17} are of the form
\begin{equation}
\mathbb{E}\left[\int_{0}^{\infty}
\overline{U}(t,\gamma_{t}c_{t})\ud\kappa_{t}\right] \to \max\,!
\label{eq:mostproblem}
\end{equation}
for some time dependent utility function
$\overline{U}(\cdot,\cdot)$. (This utility was also stochastic in
\cite{most15,chauetal17}, but this makes no difference to the argument
here.) To make the problem in \eqref{eq:mostproblem} equivalent to our
problem in \eqref{eq:vf} requires
$\overline{U}(t,\gamma_{t}c_{t})=U(c_{t})$ almost surely for all
$t\geq 0$, and this is easy to satisfy. For example, if
$\gamma_{t}=\e^{\alpha t},\,t\geq 0$ and $U(\cdot)=\log(\cdot)$ is
logarithmic utility, we choose
$\overline{U}(t,\bar{c})=\log(\bar{c})-\alpha t$. If
$U(c)=c^{p}/p,\,p<1,p\neq 0$ is power utility, then we choose
$\overline{U}(t,\bar{c})=\e^{-\alpha pt}\bar{c}^{p}/p$. Hence, we can
always restore a problem of the form in \eqref{eq:vf} (equivalent to
the problems in \cite{most15,chauetal17} up to an additive or
multiplicative constant, typically).

We choose in this work to adopt the classical definition of
consumption. Part of our reason for doing so is to make very
transparent the underlying supermartingale constraint on deflated
wealth plus cumulative deflated consumption that one must apply, if
one is to show how the program of Kramkov and Schachermayer
\cite{ks99,ks03}, suitably modified and extended, creates a natural
procedure for characterising the classical consumption duality. As
will be seen, this reveals an interesting role reversal of the primal
and dual domains in many steps of the proofs, compared with the
terminal wealth problem, because it turns out that the primal domain
in the consumption problem is $L^{1}$-bounded (with respect to a
suitable measure), but it is the dual domain that has this property in
the terminal wealth case.

\begin{remark}[Discounted units]
\label{rem:du}

There is no loss of generality in working with discounted quantities
(so in effect a zero interest rate). To see this, suppose instead that
we have a positive interest rate process $r=(r_{t})_{t\geq 0}$, so the
cash asset with initial value $1$ has positive price process
$A_{t}=\e^{\int_{0}^{t}r_{s}\ud s},\,t\geq 0$. If $\tilde{c}$ is the
un-discounted consumption process, then the problem in \eqref{eq:vf} is
$\mathbb{E}\left[
  \int_{0}^{\infty}U\left(\tilde{c}_{t}/A_{t}\right)\ud\kappa_{t}\right]
\to \max!$ We can define another utility function
$\widetilde{U}:\mathbb{R}^{2}_{+}\to\mathbb{R}$ such that
$\widetilde{U}(A_{t},\tilde{c}_{t})=U(\tilde{c}_{t}/A_{t}),\,t\geq 0$,
and the problem in \eqref{eq:vf} can then be transported to one in
terms of the raw (un-discounted) consumption rate. For example, if
$\gamma_{t}=\e^{\alpha t},\,t\geq 0$ and $U(\cdot)=\log(\cdot)$ is
logarithmic utility, we choose
$\widetilde{U}(A,\tilde{c})=\log(\tilde{c})-\log(A)$. If
$U(c)=c^{p}/p,\,p<1,p\neq 0$ is power utility, then we choose
$\widetilde{U}(A,\tilde{c})=A^{-p}\tilde{c}^{p}/p$.

\end{remark}

\begin{remark}[Stochastic utility]
\label{rem:su}

In the problem \eqref{eq:vf} we can allow $U(\cdot)$ to be stochastic,
so to also depend on $\omega\in\Omega$ in an optional way, as done by
Mostovyi \cite{most15}. The analysis is unaffected, as the reader can
easily verify, so one can read the proofs with a stochastic utility in
mind and with dependence on $\omega\in\Omega$ suppressed throughout.

\end{remark}

\subsection{The dual problem}
\label{subsec:dual}

Let $V:\mathbb{R}_{+}\to\mathbb{R}$ denote the convex conjugate of
$U(\cdot)$, defined by
\begin{equation*}
V(y) := \sup_{x>0}[U(x)-xy], \quad y>0.
\end{equation*}
The map $y\mapsto V(y),\,y>0$, is strictly convex, strictly decreasing,
continuously differentiable on $\mathbb{R}_{+}$, $-V(\cdot)$ satisfies
the Inada conditions, and we have the bi-dual relation
\begin{equation*}
U(x) := \inf_{y>0}[V(y)+xy], \quad x>0,
\end{equation*}
as well as
$V^{\prime}(\cdot)=-I(\cdot)=-(U^{\prime})^{-1}(\cdot)$,
where $I(\cdot)$ denotes the inverse of marginal utility. In
particular, we have the inequality
\begin{equation}
V(y) \geq U(x)-xy, \quad \forall\, x,y>0, \quad \mbox{with equality
iff $U^{\prime}(x)=y$}.  
\label{eq:VUbound}
\end{equation}

For each consumption deflator $Y\in\mathcal{Y}(y)$ defined in
\eqref{eq:mcY}, define a process $Y^{\gamma}$ by
\begin{equation}
Y^{\gamma}_{t} := \gamma_{t}Y_{t}, \quad t\geq 0,
\label{eq:Ybar}
\end{equation}
where $\gamma$ was defined in \eqref{eq:gamma}. For later use, denote
the set of such processes by $\widetilde{\mathcal{Y}}(y)$:
\begin{equation}
\widetilde{\mathcal{Y}}(y) := \left\{Y^{\gamma}:
\mbox{$Y^{\gamma}$ is given by \eqref{eq:Ybar}, with
  $Y\in\mathcal{Y}(y)$}\right\}, \quad y>0, 
\label{eq:Ycbar}
\end{equation}
so the set $\widetilde{\mathcal{Y}}(y)$ is in one-to-one correspondence
with the set $\mathcal{Y}(y)$ of consumption deflators. As usual, we
write $\widetilde{\mathcal{Y}}\equiv\widetilde{\mathcal{Y}}(1)$, and we
have $\widetilde{\mathcal{Y}}(y)=y\widetilde{\mathcal{Y}}$ for
$y>0$. 

The dual problem to \eqref{eq:vf} has value function
$v:\mathbb{R}_{+}\to\mathbb{R}$ defined by
\begin{equation}
v(y) := \inf_{Y\in\mathcal{Y}(y)}\mathbb{E}\left[
\int_{0}^{\infty}V(\gamma_{t}Y_{t})\ud\kappa_{t}\right], \quad
y>0.    
\label{eq:dvf}
\end{equation}
We shall assume throughout that $v(y)<\infty$ for all $y>0$.

\section{The duality theorem}
\label{sec:dt}

Here is the main result, the perpetual consumption duality. It is
somewhat stronger and mathematically more robust than previous
results. We describe how the theorem differs from, and in which senses
it strengthens, existing results, after presenting the theorem.

\begin{theorem}[Perpetual consumption duality under NUPBR]
\label{thm:consd}

Define the primal consumption problem by \eqref{eq:vf} and the
corresponding dual problem by \eqref{eq:dvf}. Assume \eqref{eq:noarb},
\eqref{eq:inada} and that
\begin{equation*}
u(x)>-\infty,\,\forall x>0, \quad v(y)<\infty,\,\forall y>0.   
\end{equation*}
Then:

\begin{itemize}

\item[(i)] $u(\cdot)$ and $v(\cdot)$ are conjugate:
\begin{equation*}
v(y) = \sup_{x>0}[u(x)-xy], \quad u(x) =
\inf_{y>0}[v(y)+xy], \quad x,y>0.
\end{equation*}

\item[(ii)] The primal and dual optimisers
$\widehat{c}(x)\in\mathcal{A}(x)$ and
$\widehat{Y}(y)\in\mathcal{Y}(y)$ exist and are unique, so that
\begin{equation*}
u(x) =
\mathbb{E}\left[\int_{0}^{\infty}U(\widehat{c}_{t}(x))\ud\kappa_{t}\right],
\quad v(y) = \mathbb{E}\left[\int_{0}^{\infty}
V(\gamma_{t}\widehat{Y}_{t}(y))\ud\kappa_{t}\right], \quad x,y>0.
\end{equation*}

\item[(iii)] With $y=u^{\prime}(x)$ (equivalently,
$x=-v^{\prime}(y)$), the primal and dual optimisers are related
by
\begin{equation}
U^{\prime}(\widehat{c}_{t}(x)) = \gamma_{t}\widehat{Y}_{t}(y), \quad
\mbox{equivalently}, \quad \widehat{c}_{t}(x) =
-V^{\prime}(\gamma_{t}\widehat{Y}_{t}(y)), \quad t\geq 0,
\label{eq:pdc}
\end{equation}
and satisfy
\begin{equation}
\mathbb{E}\left[\int_{0}^{\infty}\widehat{c}_{t}(x)\widehat{Y}_{t}(y)\ud
t\right] = xy.
\label{eq:oihbc}
\end{equation}
Moreover, the associated optimal wealth process $\widehat{X}(x)$ is
given by
\begin{equation}
\widehat{X}_{t}(x)\widehat{Y}_{t}(y)  =
\mathbb{E}\left[\left.\int_{t}^{\infty}\widehat{c}_{s}(x)\widehat{Y}_{s}(y)\ud
s\right\vert\mathcal{F}_{t}\right], \quad t\geq 0,
\label{eq:owp}
\end{equation}
and the process
$\widehat{X}(x)\widehat{Y}(y) +
\int_{0}^{\cdot}\widehat{c}_{s}(x)\widehat{Y}_{s}(y)\ud s$ is a
uniformly integrable martingale.

\item[(iv)] The functions $u(\cdot)$ and $-v(\cdot)$ are strictly
increasing, strictly concave, satisfy the Inada conditions, and for
all $x,y>0$ their derivatives satisfy
\begin{equation*}
xu^{\prime}(x) = \mathbb{E}\left[\int_{0}^{\infty}
U^{\prime}(\widehat{c}_{t}(x))\widehat{c}_{t}(x)\ud\kappa_{t}\right],
\quad yv^{\prime}(y) = \mathbb{E}\left[\int_{0}^{\infty}
V^{\prime}(\gamma_{t}\widehat{Y}_{t}(y))\widehat{Y}_{t}(y)\ud t\right].
\end{equation*}

\end{itemize}

\end{theorem}

The proof of Theorem \ref{thm:consd} will be given in Section
\ref{sec:pdt}, and will rely on bipolarity results and an abstract
version of the duality stated in Section \ref{sec:abpd}, with the
bipolarity results proven in Section \ref{sec:bcbpr}. A duality result
of this form was established by Mostovyi \cite{most15} under
NFLVR. This was strengthened to a result under NUPBR by Chau et al
\cite{chauetal17}. Compared to these papers, Theorem \ref{thm:consd}
makes a stronger statement in other ways.

First, we characterise the optimal wealth process, a statement that
was missing from \cite{most15,chauetal17}. This turns out to be a nice
result to prove (see Proposition \ref{prop:owp}), showing that the
optimal process $\widehat{X}\widehat{Y}$ is a supermartingale and a
potential, while
$\widehat{X}\widehat{Y} +
\int_{0}^{\cdot}\widehat{c}_{s}\widehat{Y}_{s}\ud s$ is a uniformly
integrable martingale. This is the natural extension of the result in
the terminal wealth problem that, at the optimum, deflated wealth is a
uniformly integrable martingale (see Kramkov and Schachermayer
\cite{ks99,ks03}), and confirms that the supermartingale condition we
placed on the process $XY+\int_{0}^{\cdot}c_{s}Y_{s}\ud s$ for
admissibility is the right criterion to start from.

Further, as we shall see in the course of proving Theorem
\ref{thm:consd}, the dual domain $\mathcal{Y}(y)$ will need to be
enlarged, in a spirit akin to Kramkov and Schachermayer
\cite{ks99,ks03}, to consider processes which are dominated by some
element of the original dual domain. This enlargement, as is known
from the terminal wealth scenario of \cite{ks99,ks03}, is needed in
order to reach the bipolar of the original dual domain, so that the
(enlarged) dual domain is closed in an appropriate topology. This in
turn guarantees that a unique dual optimiser will exist. This is one
of the key contributions made in \cite{ks99,ks03}. One does not assume
\emph{a priori} that either the primal or dual domains are closed.

Here, for the consumption problem, we shall see that we do not need to
enlarge the primal domain, only the dual domain. Mostovyi
\cite{most15} and Chau et al \cite{chauetal17} found a similar
phenomenon, but with the important caveat that they took the enlarged
dual domain to be the closure (in the appropriate topology) of the set
of processes dominated by local martingale deflators (in
\cite{chauetal17}) or martingale deflators (in \cite{most15}) .

Here, we do not explicitly make the dual domain closed (in the manner
of \cite{most15,chauetal17}) by construction, so we obtain a stronger
result. We merely enlarge the dual domain in a manner analogous to
Kramkov and Schachermayer \cite{ks99,ks03}, by considering processes
dominated by consumption deflators, and then show that the enlarged
domain is closed using supermartingale convergence results which
exploit so-called Fatou convergence of processes. We also prove that
our enlarged domain coincides with the closure of processes dominated
by local martingale deflators (see Proposition \ref{prop:dense}), so
coincides with that used in \cite{chauetal17}. In other words, the
domain used in \cite{chauetal17} is \emph{dense} in our domain. The
proof of Proposition \ref{prop:dense} will also reveal why the
supermartingale convergence results, used to show that our enlarged
domain is closed, cannot provide the same result for the (pre-closure)
domains used in \cite{most15,chauetal17}. Basically, the limiting
supermartingale is just that, a supermartingale, and it cannot be
shown to be a (local) martingale deflator.

This all reveals that, in a real sense, we have found just the right
dual domain for a strong duality statement.

Lastly, regarding some steps underlying the proof of Theorem
\ref{thm:consd}, and in particular the arguments in Section
\ref{sec:bcbpr} used to establish bipolarity relations connecting the
primal and dual domains, our proofs make no use at any point of
constructions involving equivalent measures, such as ELMMs, but use
only deflators. Since we are working under NUPBR this is natural, and
in some senses even desirable. Moreover, as we have alluded to in
Section \ref{subsubsec:completion}, there are potential complications
in using equivalent measures when working on an infinite horizon, so
there are sound reasons for taking the course we follow here.

In our scenario, therefore, we provide an unambiguously robust route
through the proofs which avoids any use of ELMMs. This, in addition to
the features described above, of showing that the naturally enlarged
dual domain is closed, without taking its closure to guarantee this,
makes Theorem \ref{thm:consd} a quite distinct infinite horizon
consumption duality result from those in \cite{most15,chauetal17}.

\begin{remark}[Incorporating a stochastic clock into the wealth
dynamics]
\label{rem:iscwd}

As discussed in Section \ref{subsec:sc}, one can incorporate a
stochastic clock into the wealth dynamics, as done by Mostovyi
\cite{most15} and Chau et al \cite{chauetal17}. Our entire program
works with this change, and we point out here how Theorem
\ref{thm:consd} would be altered. We modify the wealth dynamics
\eqref{eq:wealth} to
\begin{equation*}
X_{t} = x + (H\cdot S)_{t} - \int_{0}^{t}c_{s}\ud\kappa_{s},
\quad t\geq 0, \quad x>0,
\end{equation*} 
where $\kappa$ is a stochastic clock of the form described in Section
\ref{subsec:sc}. The process $c$ was denoted by $\bar{c}$ in Section
\ref{subsec:sc}, but for a clean notation we shall not make this
adjustment here. The primal value function is still given by
\eqref{eq:vf}. The consumption deflators are also unchanged, but the
key supermartingale constraint in \eqref{eq:mcY} is altered to reflect
the change of consumption variable, to:
\begin{equation*}
\mbox{$XY+\int_{0}c_{s}Y_{s}\ud\kappa_{s}$ is a supermartingale,
$\forall\,c\in\mathcal{A}$}, 
\end{equation*}
where admissible consumption plans are still those for which the
wealth process $X$ is non-negative. In other words, one simply alters
the measure used in the cumulative deflated consumption term in the
fundamental supermartingale constraint. As a result, the form of the
dual problem in \eqref{eq:dvf} is altered to
\begin{equation*}
v(y) := \inf_{Y\in\mathcal{Y}(y)}\mathbb{E}\left[
\int_{0}^{\infty}V(Y_{t})\ud\kappa_{t}\right], \quad
y>0,
\end{equation*}
so one loses the extraneous process $\gamma$ in the argument of
$V(\cdot)$ in the definition of the dual value function, and hence in
the expression for this function in item (ii) of the theorem. Similar
adjustments occur in the remaining results of Theorem
\ref{thm:consd}. Thus, item (iii) of the theorem is altered to:

With $y=u^{\prime}(x)$ (equivalently,
$x=-v^{\prime}(y)$), the primal and dual optimisers are related
by
\begin{equation*}
U^{\prime}(\widehat{c}_{t}(x)) = \widehat{Y}_{t}(y), \quad
\mbox{equivalently}, \quad \widehat{c}_{t}(x) =
-V^{\prime}(\widehat{Y}_{t}(y)), \quad t\geq 0,
\end{equation*}
and satisfy
\begin{equation*}
\mathbb{E}\left[\int_{0}^{\infty}\widehat{c}_{t}(x)\widehat{Y}_{t}(y)
\ud\kappa_{t}\right] = xy,
\end{equation*}
with the associated optimal wealth process $\widehat{X}(x)$ given by
\begin{equation*}
\widehat{X}_{t}(x)\widehat{Y}_{t}(y)  = \mathbb{E}\left[\left.
\int_{t}^{\infty}\widehat{c}_{s}(x)\widehat{Y}_{s}(y)\ud\kappa_{s}
\right\vert\mathcal{F}_{t}\right], \quad t\geq 0,
\end{equation*}
and the process
$\widehat{X}(x)\widehat{Y}(y) +
\int_{0}^{\cdot}\widehat{c}_{s}(x)\widehat{Y}_{s}(y)\ud\kappa_{s}$ is a
uniformly integrable martingale.

Item (iv) of the theorem is altered to:

The functions $u(\cdot)$ and $-v(\cdot)$ are strictly
increasing, strictly concave, satisfy the Inada conditions, and for
all $x,y>0$ their derivatives satisfy
\begin{equation*}
xu^{\prime}(x) = \mathbb{E}\left[\int_{0}^{\infty}
U^{\prime}(\widehat{c}_{t}(x))\widehat{c}_{t}(x)\ud\kappa_{t}\right],
\quad yv^{\prime}(y) = \mathbb{E}\left[\int_{0}^{\infty}
V^{\prime}(\widehat{Y}_{t}(y))\widehat{Y}_{t}(y)\ud\kappa_{t}\right].
\end{equation*}

\end{remark}

\section{Abstract bipolarity and duality}
\label{sec:abpd}

In this section we state a bipolarity result in abstract form, leading to an
abstract duality theorem, from which Theorem \ref{thm:consd} will
follow. Proofs of these results will follow in subsequent sections.

Set $\mathbf{\Omega}:=[0,\infty)\times\Omega$. Let $\mathcal{G}$
denote the optional $\sigma$-algebra on $\mathbf{\Omega}$, that is,
the sub-$\sigma$-algebra of
$\mathcal{B}([0,\infty))\otimes\mathcal{F}$ generated by evanescent
sets and stochastic intervals of the form
$\llbracket T,\infty\llbracket$ for arbitrary stopping times $T$.
Define the measure $\mu:=\kappa\times\mathbb{P}$ on
$(\mathbf{\Omega},\mathcal{G})$. On the resulting finite measure space
$(\mathbf{\Omega},\mathcal{G},\mu)$, denote by $L^{0}_{+}(\mu)$ the
space of non-negative $\mu$-measurable functions, corresponding to
non-negative infinite horizon processes.

The primal and dual domains for our optimisation problems
\eqref{eq:vf} and \eqref{eq:dvf} are now considered as subsets of
$L^{0}_{+}(\mu)$. The abstract primal domain $\mathcal{C}(x)$ is
identical to the set of admissible consumption plans, now considered
as a subset of $L^{0}_{+}(\mu)$:
\begin{equation}
\mathcal{C}(x) := \{g\in L^{0}_{+}(\mu): \mbox{$g=c,\,\mu$-a.e., for
some $c\in\mathcal{A}(x)$}\}, \quad x>0.  
\label{eq:Cx}
\end{equation}
As always we write $\mathcal{C}\equiv\mathcal{C}(1)$, with
$\mathcal{C}(x)=x\mathcal{C}$ for $x>0$, and the set $\mathcal{C}$ is
convex. (Since $\mathcal{C}=\mathcal{A}$ we do not really need to
introduce the new notation, and do so only for some notational
symmetry in the abstract formulation.) In the abstract notation, the
primal value function \eqref{eq:vf} is written as
\begin{equation}
u(x) := \sup_{g\in\mathcal{C}(x)}\int_{\mathbf{\Omega}}U(g)\ud\mu,
\quad x>0.
\label{eq:vfabs}
\end{equation}

For the dual problem, the abstract dual domain is an enlargement
of the original domain to accommodate processes dominated by the
original dual variables. To this end, define the set
\begin{equation}
\mathcal{D}(y) := \{h\in L^{0}_{+}(\mu): \mbox{$h\leq \gamma
  Y,\,\mu$-a.e., for some $Y\in\mathcal{Y}(y)$}\}, \quad y>0.
\label{eq:Dy}
\end{equation}
As usual, we write $\mathcal{D}\equiv\mathcal{D}(1)$, we have
$\mathcal{D}(y)=y\mathcal{D}$ for $y>0$, and the set $\mathcal{D}$ is
convex. With this notation, and since $V(\cdot)$ is decreasing, the
dual problem \eqref{eq:dvf} takes the form
\begin{equation}
v(y) := \inf_{h\in\mathcal{D}(y)}\int_{\mathbf{\Omega}}V(h)\ud\mu,
\quad y>0. 
\label{eq:dvfabs}
\end{equation}

The enlargement of the dual domain from $\mathcal{Y}$ (equivalently,
$\widetilde{\mathcal{Y}}$ in \eqref{eq:Ycbar}) to $\mathcal{D}$ is
needed for the same reason as in Kramkov and Schachermayer
\cite{ks99,ks03} in the context of the terminal wealth problem (where
one enlarged the dual domain from supermartingale deflators to
elements of $L^{0}_{+}(\mathbb{P})$ that were dominated by terminal
values of supermartingale deflators). The enlargement will ensure that
$\mathcal{D}$ is closed with respect to convergence in measure $\mu$
(proven in Lemma \ref{lem:Dclosed}). This in turn ensures that we
reach a perfect bipolarity between the primal and dual domains (as
given in Proposition \ref{prop:abp}), which is a key ingredient in
establishing full duality between the primal and dual
problems. Contrast this enlargement with the approach taken in Chau et
al \cite{chauetal17} and Mostovyi \cite{most15} as described
immediately below.

\subsection{Alternative dual domains}
\label{subsec:add}

In Chau et al \cite{chauetal17} (respectively, Mosotvyi \cite{most15})
the dual domain was not based on the deflators $Y\in\mathcal{Y}$ but
instead on the local martingale deflators $Z\in\mathcal{Z}$
(respectively, equivalent martingale deflators). Thus, translated into
our formulation (so using a true rather than a pseudo-consumption
rate), Chau et al \cite{chauetal17} use, in place of $\mathcal{D}(y)$,
a domain defined as the closure, with respect to
the topology of convergence in measure $\mu$, of a set $D(y)$, where
$D(y)$ is defined analogously to $\mathcal{D}(y)$ but with local
martingale deflators replacing the consumption deflators. Thus, with
$\overline{A}\equiv\cl(A)$ denoting the closure of any set
$A\subseteq L^{0}_{+}(\mu)$, we have
\begin{equation}
\overline{D}(y) \equiv \cl(D) := \cl\left\{h\in
L^{0}_{+}(\mu):\,h\leq y\gamma Z, \,\mbox{for some
$Z\in\mathcal{Z}$}\right\}, \quad y>0.
\label{eq:wtDy}
\end{equation}
As usual we write $D\equiv D(1)$, and $D(y)=yD$ for $y>0$, with the
same convention for $\overline{D}$. In this formulation,
therefore, the dual value function is represented as in
\eqref{eq:dvfabs} but with $\overline{D}(y)$ in place of
$\mathcal{D}(y)$.

The salient point here is the fact that the \emph{closure} of $D(y)$
has been taken in \eqref{eq:wtDy}. The reason for this will become
transparent in the proofs of Section \ref{sec:bcbpr}, but we outline
the issue here, and state a nice result (Proposition \ref{prop:dense})
which connects the domains $\mathcal{D}$, $D$ and $\overline{D}$.

In the approach of \cite{chauetal17} (and also of \cite{most15}, with
martingale deflators in place of local martingale deflators), if one
does not take the aforementioned closure, it becomes impossible (as
far as we can see) to prove that the dual domain is closed. It thus
becomes impossible to obtain a perfect bipolarity between the primal
and dual domains, on which the duality proofs ultimately rest. The
technical reason for this is that the closed property of $\mathcal{D}$
is established (see Lemma \ref{lem:Dclosed}) using a supermartingale
convergence result based on Fatou convergence of processes. The
limiting supermartingale in this procedure is known only to be a
supermartingale in $\mathcal{Y}$, so is not guaranteed to be a local
martingale deflator. This is the driving force behind our choice of
dual domain based on a supermartingale criterion. The approach in
\cite{chauetal17,most15} is simply not amenable to this procedure,
which is why those papers had to invoke the closure in
\eqref{eq:wtDy}.

In this way, we strengthen the duality theorems in
\cite{most15,chauetal17}, by not forcing the dual domain to be closed
by construction. This point is well made by Rogers \cite{rogers}, who
observes that having to take the closure of the dual domain in its
definition ``makes the statement of the main result somewhat weaker''.
We do denigrate in any way, however, the advances made in
\cite{most15,chauetal17}.

What is more, we have the proposition below, which reaffirms in some
sense that our choice of dual domain is the correct one: we have
chosen it in just the right way to reach the bipolar of the original
dual domain and hence the polar of the primal domain.

\begin{proposition}
\label{prop:dense}

With respect to the topology of convergence in measure $\mu$, the set
\begin{equation*}
D := \left\{h\in L^{0}_{+}(\mu):\,h\leq \gamma Z, \,\mbox{for some
$Z\in\mathcal{Z}$}\right\}, 
\end{equation*}
is dense in the set $\mathcal{D}\equiv\mathcal{D}(1)$ of
\eqref{eq:Dy}. That is, we have
\begin{equation*}
\mathcal{D} = \overline{D} \equiv \cl(D).
\end{equation*}

\end{proposition}

The proof of Proposition \ref{prop:dense} will be given in Section
\ref{sec:bcbpr}, alongside the proof of the bipolarity result in
Proposition \ref{prop:abp} that is the subject of the next subsection.

\subsection{Abstract bipolarity}
\label{subsec:absbp}

The abstract duality theorem relies on the abstract bipolarity result
in Proposition \ref{prop:abp} below which connects the sets
$\mathcal{C}$ and $\mathcal{D}$. The result is of course in the spirit
of Kramkov and Schachermayer \cite[Proposition 3.1]{ks99}.

We shall sometimes employ the notation
\begin{equation*}
\langle g,h\rangle := \int_{\mathbf{\Omega}}gh\ud\mu, \quad
g,h\in L^{0}_{+} (\mu).
\end{equation*}

Let us recall some definitions, particularly the concepts of set
solidity and the polar of a set.

\begin{definition}[Solid set, closed set]
\label{def:sscs}
  
A subset $A\subseteq L^{0}_{+}(\mu)$ is called \emph{solid} if $f\in A$
and $0\leq g\leq f,\,\mu$-a.e. implies that $g\in A$. 

A set is \emph{closed in $\mu$-measure}, or simply {\em closed}, if it
is closed with respect to the topology of convergence in measure
$\mu$.

\end{definition}

\begin{definition}[Polar of a set]
  
The {\em polar}, $A^{\circ}$, of a set $A\subseteq L^{0}_{+}(\mu)$,
is defined by
\begin{equation*}
A^{\circ} := \left\{h\in L^{0}_{+}(\mu): \langle g,h\rangle\leq
1,\,\mbox{for each $g\in A$}\right\}. 
\end{equation*}

\end{definition}

For clarity and for later use, we state here the bipolar theorem of
Brannath and Schachermayer \cite[Theorem 1.3]{bs99}, originally proven
in a probability space, and adapted here to the measure space
$(\mathbf{\Omega},\mathcal{G},\mu)$.

\begin{theorem}[Bipolar theorem, Brannath and Schachermayer
\cite{bs99}, Theorem 1.3]
\label{thm:bp}

On the finite measure space $(\mathbf{\Omega},\mathcal{G},\mu)$:

\begin{itemize}

\item[(i)] For a set $A\subseteq L^{0}_{+}(\mu)$, its polar
  $A^{\circ}$ is a closed, convex, solid subset of $L^{0}_{+}(\mu)$.

\item[(ii)] The bipolar $A^{\circ\circ}$, defined by
\begin{equation*}
A^{\circ\circ} := \left\{g\in L^{0}_{+}(\mu): \langle g,h\rangle\leq
1,\,\mbox{for each $h\in A^{\circ}$}\right\},
\end{equation*}
is the smallest closed, convex, solid set in $L^{0}_{+}(\mu)$
containing $A$.
  
\end{itemize}

\end{theorem}

\begin{proposition}[Abstract bipolarity]
\label{prop:abp}

Under the condition \eqref{eq:noarb}, the abstract primal and dual
sets $\mathcal{C}$ and $\mathcal{D}$ satisfy the following properties:

\begin{itemize}

\item[(i)] $\mathcal{C}$ and $\mathcal{D}$ are both closed with
respect to convergence in measure $\mu$, convex and solid;

\item[(ii)] $\mathcal{C}$ and $\mathcal{D}$ satisfy the bipolarity
relations
\begin{eqnarray}
g\in\mathcal{C} & \iff & \langle g,h\rangle \leq 1, \quad \forall\,
h\in\mathcal{D}, \quad \mbox{that is,
$\mathcal{C}=\mathcal{D}^{\circ}$}, \label{eq:bp1}\\
h\in\mathcal{D} & \iff & \langle g,h\rangle \leq 1, \quad \forall\,
g\in\mathcal{C}, \quad \mbox{that is,
$\mathcal{D}=\mathcal{C}^{\circ}$};  \label{eq:bp2}
\end{eqnarray}

\item[(iii)] $\mathcal{C}$ and $\mathcal{D}$ are bounded in
$L^{0}(\mu)$, and $\mathcal{C}$ is also bounded in $L^{1}(\mu)$.

\end{itemize}

\end{proposition}

The proof of Proposition \ref{prop:abp} will be given in Section
\ref{sec:bcbpr}, where we shall establish the infinite horizon budget
constraint, giving a necessary condition for admissible consumption
plans, and a reverse implication, leading to a sufficient condition
for admissibility, culminating in the full bipolarity relations once
we enlarge the dual domain. The derivations in Section \ref{sec:bcbpr}
are quite distinct from previous approaches, and are the bedrock of
the mathematical results. As indicated earlier, we shall establish the
bipolarity results without any recourse whatsoever to constructions
involving ELMMs, by exploiting ramifications of the Stricker and Yan
\cite{sy98} version of the optional decomposition theorem.

\subsection{Abstract duality}
\label{subsec:ad}

Armed with the abstract bipolarity in Proposition \ref{prop:abp}, we
have the following abstract version of the convex duality relations
between the primal problem \eqref{eq:vfabs} and its dual
\eqref{eq:dvfabs}. The theorem shows that all the natural tenets of
utility maximisation theory, as established by Kramkov and
Schachermayer \cite{ks99} in the terminal wealth problem under NFLVR,
extend to infinite horizon inter-temporal problems under NUPBR, with
weak underlying assumptions on the primal and dual domains.

\begin{theorem}[Abstract duality theorem]
\label{thm:adt}

Define the primal value function $u(\cdot)$ by \eqref{eq:vfabs} and the
dual value function by  \eqref{eq:dvfabs}. Assume that the utility function
satisfies the Inada conditions \eqref{eq:inada} and that
\begin{equation}
u(x) > -\infty,\,\forall \, x>0, \quad v(y)< \infty,\,\forall \,y>0.   
\label{eq:minimal}
\end{equation}

Then, with Proposition \ref{prop:abp} in place, we have:

\begin{itemize}

\item[(i)] $u(\cdot)$ and $v(\cdot)$ are conjugate:
\begin{equation}
v(y) = \sup_{x>0}[u(x)-xy], \quad u(x) = \inf_{y>0}[v(y)+xy], \quad
x,y>0.
\label{eq:conjugacy}
\end{equation}

\item[(ii)] The primal and dual optimisers
$\widehat{g}(x)\in\mathcal{C}(x)$ and
$\widehat{h}(y)\in\mathcal{D}(y)$ exist and are unique, so that
\begin{equation*}
u(x) = \int_{\mathbf{\Omega}}U(\widehat{g}(x))\ud\mu, \quad v(y) =
\int_{\mathbf{\Omega}}V(\widehat{h}(y))\ud\mu, \quad x,y>0.  
\end{equation*}

\item[(iii)] With $y=u^{\prime}(x)$ (equivalently,
$x=-v^{\prime}(y)$), the primal and dual optimisers are related by
\begin{equation*}
U^{\prime}(\widehat{g}(x)) = \widehat{h}(y), \quad
\mbox{equivalently}, \quad \widehat{g}(x) =
-V^{\prime}(\widehat{h}(y)), 
\end{equation*}
and satisfy
\begin{equation*}
\langle \widehat{g}(x),\widehat{h}(y)\rangle = xy.
\end{equation*}

\item[(iv)] $u(\cdot)$ and $-v(\cdot)$ are strictly increasing,
strictly concave, satisfy the Inada conditions, and their
derivatives satisfy
\begin{equation*}
xu^{\prime}(x) = \int_{\mathbf{\Omega}}
U^{\prime}(\widehat{g}(x))\widehat{g}(x)\ud\mu, \quad
yv^{\prime}(y) = \int_{\mathbf{\Omega}}
V^{\prime}(\widehat{h}(y))\widehat{h}(y)\ud\mu, \quad
x,y>0. 
\end{equation*}

\end{itemize}

\end{theorem}

The proof of Theorem \ref{thm:adt} will be given in Section
\ref{sec:pdt}, and uses as its starting point the bipolarity result in
Proposition \ref{prop:abp}. 

The duality proof itself follows some of the classical steps (with
adaptations) of Kramkov and Schachermayer \cite{ks99,ks03}, but there
is an interesting role reversal for the primal and dual sets. In the
terminal wealth problem, the dual domain is bounded in
$L^{1}(\mathbb{P})$, because the constant wealth process
$\mathbbm{1}:\Omega\mapsto 1$ lies in the primal domain. In the
infinite horizon consumption problem, by contrast, the constant
consumption stream $c\equiv 1$ is not admissible, so the dual domain
is not bounded in $L^{1}(\mu)$. Instead, it turns out that
$L^{1}(\mu)$-boundedness is satisfied by the primal domain. The upshot
is that, in a number of places, the method of proof used in
\cite{ks99,ks03} for a property of the primal domain is applied in our
case to a corresponding property in the dual domain (and vice
versa). Examples include the proofs of uniform integrability of the
families $(U^{+}(g))_{g\in\mathcal{C}(x)}$ and
$(V^{-}(h))_{h\in\mathcal{D}(y)}$, a reversed application of the
minimax theorem (replacing a maximisation with a minimisation and so
forth) in proving conjugacy of the value functions, and some
characterisations of the derivatives of the value functions at zero
and infinity. We shall point out these features when proving the
results. This is one of the reasons for our choosing to give a
complete, self-contained treatment with full proofs.

We conclude this section with a small remark (that is by now standard,
but does need stating) on reasonable asymptotic elasticity as an
alternative to assuming finiteness of the dual value function.

\begin{remark}[Reasonable asymptotic elasticity]
\label{rem:aeu}
  
In Theorem \ref{thm:adt} we have assumed only the minimal conditions
in \eqref{eq:minimal} to guarantee non-trivial primal and dual
problems. It is well-known that, in place of the second condition in
\eqref{eq:minimal} of a finitely-valued dual problem, we could have
imposed the reasonable asymptotic elasticity condition of Kramkov and
Schachermayer \cite{ks99} as given in \eqref{eq:AEU}, along with the
assumption that $u(x)<\infty$ for some $x>0$. Then, as in Kramkov and
Schachermayer \cite[Note 2]{ks03}, these conditions would have implied
that $v(y)<\infty$ for all $y>0$.

\end{remark}

\section{Budget constraint and bipolarity relations}
\label{sec:bcbpr}

\subsection{The budget constraint}
\label{subsec:bc}

The first step in the proof of the duality theorem is to establish
bipolarity relations between the primal and dual domains. We shall do
this in stages, first deriving the infinite horizon budget
constraint. This yields the form of the dual problem as a
byproduct. The derivation also lends itself to a discussion of the
rationale for choosing the dual domain to be the set $\mathcal{Y}(y)$
of consumption deflators, and what would have been the ramifications
of instead choosing the wealth deflators or the local martingale
deflators as the dual variables.

\begin{lemma}[Infinite horizon budget constraint]
\label{lem:ihbc}

Let $c\in\mathcal{A}(x)$ be any admissible consumption plan and let
$Y\in\mathcal{Y}(y)$ be any consumption deflator. We then have
the \emph{infinite horizon budget constraint}:
\begin{equation}
\mathbb{E}\left[\int_{0}^{\infty}c_{t}Y_{t}\ud t\right] \leq xy, \quad
\forall\,c\in\mathcal{A}(x),\,Y\in\mathcal{Y}(y).
\label{eq:ihbc}
\end{equation}

\end{lemma}

\begin{proof}

Recall the wealth process $X$ incorporating consumption in
\eqref{eq:wealth}. Since
$XY + \int_{0}^{\cdot}c_{s}Y_{s}\ud s$ is a supermartingale and
$XY\geq 0$, we have 
\begin{equation*}
\mathbb{E}\left[\int_{0}^{t}c_{s}Y_{s}\ud s\right] \leq xy, \quad
t\geq 0.  
\end{equation*}
Letting $t\uparrow\infty$ and using monotone convergence we
obtain \eqref{eq:ihbc}.

\end{proof}

\begin{remark}[On alternative choices of dual domain]
\label{rem:alternatives}

The derivation of Lemma \ref{lem:ihbc} allows us to give some of the
rationale for choosing the dual domain as we did.

Suppose instead that we chose the dual domain to be the set
$\mathcal{Y}^{0}(y)$ of supermartingale deflators. Recall the
decomposition in \eqref{eq:wdecomp} of a wealth process $X$
incorporating consumption into a self-financing wealth process $X^{0}$
minus cumulative consumption $C=\int_{0}^{\cdot}c_{s}\ud s$. Now, for
any wealth deflator $Y^{0}\in\mathcal{Y}^{0}(y)$ and
$c\in\mathcal{A}(x)$ we have, on using the It\^o product rule on the
process $CY^{0}$ and re-arranging,
\begin{equation}
XY^{0} + \int_{0}^{\cdot}c_{s}Y^{0}_{s}\ud s =  X^{0}Y^{0} -
\int_{0}^{\cdot}C_{s-}\ud Y^{0}_{s}.
\label{eq:XcXY}
\end{equation}
The right-hand-side of \eqref{eq:XcXY} is a difference of
supermartingales, so not necessarily a supermartingale, and we would
fail to achieve the infinite horizon budget constraint.

Suppose, on the other hand, that we chose the dual domain to be
constructed from the set $\mathcal{Z}$ of local martingale
deflators. This is the route taken by Chau et al \cite{chauetal17} and
by Mostovyi \cite{most15} (except that the deflators were martingales
in \cite{most15}, in tandem with the NFLVR scenario in that paper.)
We would then reach the analogue of \eqref{eq:XcXY} in the form
\begin{equation*}
XZ + \int_{0}^{\cdot}c_{s}Z_{s}\ud s =  X^{0}Z -
\int_{0}^{\cdot}C_{s-}\ud Z_{s},
\end{equation*}
for any $Z\in\mathcal{Z}$. Now, $X^{0}Z$ is a non-negative local
martingale and thus a supermartingale, so using this and that
$XZ\geq 0$, we would obtain
\begin{equation}
\mathbb{E}\left[\int_{0}^{t}c_{s}Z_{s}\ud s\right] \leq x -
\mathbb{E}\left[\int_{0}^{t}C_{s-}\ud Z_{s}\right], \quad
t\geq 0.
\label{eq:EcZ}  
\end{equation}
The process $M:=\int_{0}^{\cdot}C_{s-}\ud Z_{s}$ is a local
martingale. With $(T_{n})_{n\in\mathbb{N}}$ a localising sequence for
$M$, so that
$\mathbb{E}\left[\int_{0}^{T_{n}}C_{s-}\ud
  Z_{s}\right]=0,\,n\in\mathbb{N}$, \eqref{eq:EcZ} would convert to
$\mathbb{E}\left[\int_{0}^{T_{n}}c_{s}Z_{s}\ud s\right] \leq
x,\,n\in\mathbb{N}$. Letting $n\uparrow\infty$ and using monotone
convergence we would obtain a budget constraint
$\mathbb{E}\left[\int_{0}^{\infty}c_{t}Z_{t}\ud t\right]\leq x$. So
far so good. The difficulty in taking this route would arise later,
when enlarging the dual domain to try to reach its bipolar. One seeks
to enlarge the dual domain to processes which are dominated by some
process in the original dual domain, and then to show that the
enlarged domain is closed with respect to convergence in measure
$\mu$. The closedness proof relies on exploiting Fatou convergence of
supermartingales. The limit in this procedure is known to be a
supermartingale, but there is no guarantee that it is a local
martingale deflator. So duality would ultimately fail, unless the
enlarged dual domain was made closed by explicit construction. This is
why Mostovyi \cite{most15} (respectively, Chau et al
\cite{chauetal17}) used a construction of the form in \eqref{eq:wtDy},
invoking the closure. Ultimately, as stated in Proposition
\ref{prop:dense}, all avenues reach the same goal, but the difference
is that in our approach we did not have to invoke a closure. We shall
return to this discussion of dual domains and their relations in
Remark \ref{rem:relbp}, once we have established full bipolarity
between our abstract primal and dual domains.

\end{remark}

From Lemma \ref{lem:ihbc} we obtain the form of the dual problem to
\eqref{eq:vf} by bounding the achievable utility in the familiar way.
For any $c\in\mathcal{A}(x)$ and $Y\in\mathcal{Y}(y)$ we have
\begin{eqnarray*}
\mathbb{E}\left[\int_{0}^{\infty}U(c_{t})\ud\kappa_{t}\right] & \leq & 
\mathbb{E}\left[\int_{0}^{\infty}U(c_{t})\ud\kappa_{t}\right] +
xy - \mathbb{E}\left[\int_{0}^{\infty}c_{t}Y_{t}\ud t\right] \\
& = & \mathbb{E}\left[\int_{0}^{\infty}\left(U(c_{t}) -
c_{t}\gamma_{t}Y_{t}\right)\ud\kappa_{t}\right] + xy \\ 
& \leq & \mathbb{E}\left[\int_{0}^{\infty}
V\left(\gamma_{t}Y_{t}\right)\ud\kappa_{t}\right] + xy, \quad x,y>0,  
\end{eqnarray*}
the last inequality a consequence of \eqref{eq:VUbound}. This
motivates the definition of the dual problem associated with the
primal problem \eqref{eq:vf}, with dual value function $v(\cdot)$
defined by \eqref{eq:dvf}.

\subsection{Bipolar relations}
\label{subsec:bpr}

In economic terms, the budget constraint \eqref{eq:ihbc} says that
initial capital can finance future consumption, and constitutes a
necessary condition for admissible consumption processes. Indeed,
another way of defining admissible consumption plans is to insist
that, at any time $t\geq 0$, current wealth (suitably deflated) must
finance future deflated consumption. We would thus require
\begin{equation*}
X_{t}Y_{t} \geq \mathbb{E}\left[\left.\int_{t}^{\infty}c_{s}Y_{s}\ud
s\right\vert\mathcal{F}_{t}\right], \quad t\geq 0,
\end{equation*}
for all deflators $Y\in\mathcal{Y}(y)$. Re-arranging the above
inequality, we have
\begin{equation*}
\mathbb{E}\left[\left.\int_{0}^{\infty}c_{s}Y_{s}\ud
s\right\vert\mathcal{F}_{t}\right] \leq  X_{t}Y_{t} +
\int_{0}^{t}c_{s}Y_{s}\ud s, \quad t\geq 0.
\end{equation*}
Taking expectations, one recovers the infinite horizon budget
constraint provided that the supermartingale condition in
\eqref{eq:mcY} holds. This is another justification for the choice of
dual domain as we have presented it.

Setting $x=y=1$ in \eqref{eq:ihbc}, the budget constraint gives us
that, for $c\in\mathcal{A}$ and $Y\in\mathcal{Y}$, we have
$\mathbb{E}\left[\int_{0}^{\infty}c_{t}Y_{t}\ud t\right]\leq 1$. We
thus have the implications
\begin{equation}
c\in\mathcal{A} \implies
\mathbb{E}\left[\int_{0}^{\infty}c_{t}Y_{t}\ud t\right]\leq 1, \quad
\forall\,Y\in\mathcal{Y},
\label{eq:forwardc}
\end{equation}
and
\begin{equation}
Y\in\mathcal{Y} \implies
\mathbb{E}\left[\int_{0}^{\infty}c_{t}Y_{t}\ud t\right]\leq 1, \quad
\forall\,c\in\mathcal{A}.
\label{eq:forwardY}
\end{equation}

We wish to establish the reverse implications in some form, if need be
by enlarging the domains. First, we establish the reverse implication
to \eqref{eq:forwardc} in Lemma \ref{lem:suffc} below. This requires
some version of the Optional Decomposition Theorem (ODT), whose
original form is due to El Karoui and Quenez \cite{ekq95} in a
Brownian setting. This was generalised to the locally bounded
semimartingale case by Kramkov \cite{k96} , extended to the
non-locally bounded case by F\"ollmer and Kabanov \cite{fkab98}, and
to models with constraints by F\"ollmer and Kramkov \cite{fk97}.

The relevant version of the ODT for us is the one due to Stricker and
Yan \cite{sy98}, which uses \emph{deflators} (and in particular LMDs)
rather then ELMMs. In the proof of Lemma \ref{lem:suffc} we shall
apply a part of the Stricker and Yan ODT which applies to the
super-hedging of American claims, so is designed to construct a
process which can super-replicate a payoff at an arbitrary time. The
salient observation is that this result can also be used to dominate a
consumption stream, which is how we shall employ it. For clarity and
convenience of the reader, we state here the ODT results we need, and
afterwards specify precisely which results from \cite{sy98} we have
taken.

For $t\geq 0$, let $\mathcal{T}(t)$ denote the set of
$\mathbb{F}$-stopping times with values in $[t,\infty)$. For $t=0$,
write $\mathcal{T}\equiv\mathcal{T}(0)$, and recall the set
$\mathcal{Z}$ of local martingale deflators in \eqref{eq:mcZ}.

\begin{theorem}[Stricker and Yan \cite{sy98} ODT]
\label{thm:syodt}

\begin{itemize}
  
\item[(i)] Let $W$ be an adapted non-negative process. The process
$ZW$ is a supermartingale for each $Z\in\mathcal{Z}$ if and only if
$W$ admits a decomposition of the form
\begin{equation*}
W = W_{0} + (\phi\cdot S) - A,
\end{equation*}
where $\phi$ is a predictable $S$-integrable process such that
$Z(\phi\cdot S)$ is a local martingale for each $Z\in\mathcal{Z}$, $A$
is an adapted increasing process with $A_{0}=0$, and for all
$Z\in\mathcal{Z}$ and $T\in\mathcal{T}$,
$\mathbb{E}[Z_{T}A_{T}]<\infty$. In this case, moreover, we have
$\sup_{Z\in\mathcal{Z},T\in\mathcal{T}}\mathbb{E}[Z_{T}A_{T}]\leq
W_{0}$.

\item[(ii)] Let $b=(b_{t})_{t\geq 0}$ be a non-negative c\`adl\`ag
process such that
$\sup_{Z\in\mathcal{Z},T\in\mathcal{T}}\mathbb{E}[Z_{T}b_{T}]<\infty$. Then
there exists an adapted c\`adl\`ag process $W$ that dominates $b$:
$W_{t}\geq b_{t}$ almost surely for all $t\geq 0$, $ZW$ is a
supermartingale for each $Z\in\mathcal{Z}$, and the smallest such
process $W$ is given by
\begin{equation}
W_{t} = \esssup_{Z\in\mathcal{Z},T\in\mathcal{T}(t)}\frac{1}{Z_{t}}
\mathbb{E}[Z_{T}b_{T}\vert\mathcal{F}_{t}], \quad t\geq 0. 
\label{eq:syam}
\end{equation}

\end{itemize}

\end{theorem}

Part (i) of Theorem \ref{thm:syodt} is taken from \cite[Theorem
2.1]{sy98}. Part (ii) is a combination of \cite[Lemma 2.4 and Remark
2]{sy98}. 

The following lemma establishes the reverse implication to
\eqref{eq:forwardc}.

\begin{lemma}
\label{lem:suffc}
  
Suppose $c$ is a non-negative adapted c\`adl\`ag process that
satisfies, for all $Y\in\mathcal{Y}$,
\begin{equation}
\mathbb{E}\left[\int_{0}^{\infty}c_{t}Y_{t}\ud t\right]\leq 1.
\label{eq:ecz}
\end{equation}
Then $c\in\mathcal{A}$.

\end{lemma}

\begin{proof}

Since $c$ is assumed to satisfy \eqref{eq:ecz} for all deflators
$Y\in\mathcal{Y}$, and since $\mathcal{Z}\subseteq\mathcal{Y}$,
\eqref{eq:ecz} is satisfied for any $Z\in\mathcal{Z}$. For such a
local martingale deflator, and for any stopping time
$T\in\mathcal{T}$, the integration by parts formula gives
\begin{equation}
C_{T}Z_{T} = \int_{0}^{T}C_{s-}\ud Z_{s}+
\int_{0}^{T}c_{s}Z_{s}\ud s, \quad T\in\mathcal{T},  
\label{eq:ibpzoc}
\end{equation}
where $C:=\int_{0}^{\cdot}c_{s}\ud s$ is the non-decreasing candidate
cumulative consumption process. The process
$M:=\int_{0}^{\cdot}C_{s-}\ud Z_{s}$ is a local martingale. Let
$(T_{n})_{n\in\mathbb{N}}$ be a localising sequence for $M$, an almost
surely increasing sequence of stopping times with
$\lim_{n\to\infty}T_{n}=\infty$ a.s. such that the stopped process
$M^{T_{n}}_{t}:= M_{t\wedge T_{n}},\,t\geq 0$ is a uniformly
integrable martingale for each $n\in\mathbb{N}$. Therefore,
$\mathbb{E}\left[\int_{0}^{T\wedge T_{n}}C_{s-}\ud Z_{s}\right]=0$ for
each $n\in\mathbb{N}$. Using this along with the finiteness of
$T\in\mathcal{T}$ and the uniform integrability of $M^{T_{n}}$, we
have
\begin{equation*}
\mathbb{E}\left[\int_{0}^{T}C_{s-}\ud Z_{s}\right] =
\mathbb{E}\left[\lim_{n\to\infty}\int_{0}^{T\wedge
T_{n}}C_{s-}\ud Z_{s}\right]  =
\lim_{n\to\infty}\mathbb{E}\left[\int_{0}^{T\wedge
T_{n}}C_{s-}\ud Z_{s}\right] = 0.
\end{equation*}
Using this in \eqref{eq:ibpzoc} we obtain
\begin{equation*}
\mathbb{E}[Z_{T}C_{T}] = \mathbb{E}\left[\int_{0}^{T}Z_{s}c_{s}\ud
  s\right] \leq 1,  
\end{equation*}
the last inequality a consequence of the assumption
\eqref{eq:ecz} and $\mathcal{Z}\subseteq\mathcal{Y}$. Since
$Z\in\mathcal{Z}$ and $T\in\mathcal{T}$ were arbitrary, we have
\begin{equation*}
\sup_{Z\in\mathcal{Z},T\in\mathcal{T}}\mathbb{E}[Z_{T}C_{T}]\leq
1 < \infty.
\end{equation*}
Thus, from part (ii) of Theorem \ref{thm:syodt}, there exists a
c\`adl\`ag process $W$ that dominates $C$, so
$W_{t}\geq C_{t},\,\mathrm{a.s.},\,\forall t\geq 0$, and
$ZW$ is a super-martingale for each $Z\in\mathcal{Z}$. From
\eqref{eq:syam}, the smallest such $W$ given by
\begin{equation*}
W_{t} = \esssup_{Z\in\mathcal{Z},T\in\mathcal{T}(t)}\frac{1}{Z_{t}}
\mathbb{E}[Z_{T}C_{T}\vert\mathcal{F}_{t}], \quad t\geq 0,
\end{equation*}
so that $W_{0}\leq 1$. Further, by part (i) of Theorem
\eqref{thm:syodt}, there exists a predictable $S$-integrable process
$H$ and an adapted increasing process $A$, with $A_{0}=0$, such that
$W$ has decomposition $W=W_{0}+(H\cdot S)-A$, with $Z(H\cdot S)$ a
local martingale for each $Z\in\mathcal{Z}$, and
$\mathbb{E}[Z_{T}A_{T}]<\infty$ for all $Z\in\mathcal{Z}$ and
$T\in\mathcal{T}$.

Since $W$ dominates $C$, we can define a process $X^{0}$ by
\begin{equation*}
X^{0}_{t} := 1 + (H\cdot S)_{t}, \quad t\geq 0,
\end{equation*}
which also dominates $C$, since its initial value is no smaller than
$W_{0}$ and we have dispensed with the increasing process $A$. We
observe that $X^{0}$ corresponds to the value of a self-financing
wealth process with initial capital $1$ which dominates $C$, so that
$c\in\mathcal{A}$.

\end{proof}

We can now assemble consequences of the budget constraint and of Lemma
\ref{lem:suffc} which, combined with the bipolar theorem, gives the
following polarity properties of the set $\mathcal{A}$.

\begin{lemma}[Polarity properties of $\mathcal{A}$]
\label{lem:Aprop}

The set $\mathcal{A}\equiv\mathcal{A}(1)$ of admissible consumption
plans with initial capital $x=1$ is a closed, convex and solid subset
of $L^{0}_{+}(\mu)$. It is equal to the polar of the set
$\widetilde{\mathcal{Y}}\equiv\widetilde{\mathcal{Y}}(1)$ of
\eqref{eq:Ycbar} with respect to measure $\mu$:
\begin{equation}
\mathcal{A} = \widetilde{\mathcal{Y}}^{\circ},
\label{eq:AY0}
\end{equation}
so that
\begin{equation}
\mathcal{A}^{\circ} = \widetilde{\mathcal{Y}}^{\circ\circ},
\label{eq:AY00}
\end{equation}
and $\mathcal{A}$ is equal to its bipolar:
\begin{equation}
\mathcal{A}^{\circ\circ} = \mathcal{A}.
\label{eq:Abipolar}
\end{equation}

\end{lemma}

\begin{proof}

Lemma \ref{lem:suffc}, combined with the implication in
\eqref{eq:forwardc}, gives the equivalence
\begin{equation*}
c\in\mathcal{A} \iff
\mathbb{E}\left[\int_{0}^{\infty}c_{t}Y_{t}\ud t\right]\leq 1, \quad
\forall\,Y\in\mathcal{Y}.  
\end{equation*}
In terms of the measure $\kappa$ of \eqref{eq:kappa} and the set
$\widetilde{\mathcal{Y}}$ in \eqref{eq:Ycbar} of processes
$\gamma Y,\,Y\in\mathcal{Y}$, we have
\begin{equation*}
c\in\mathcal{A} \iff
\mathbb{E}\left[\int_{0}^{\infty}c_{t}Y^{\gamma}_{t}\ud\kappa_{t}\right]\leq
1, \quad \forall\,Y^{\gamma}\in\widetilde{\mathcal{Y}}.
\end{equation*}
Equivalently, in terms of the measure $\mu$, we have
\begin{equation}
c\in\mathcal{A} \iff 
\int_{\mathbf{\Omega}}cY^{\gamma}\ud\mu \leq 1, \quad \forall\,
Y^{\gamma}\in\widetilde{\mathcal{Y}}.
\label{eq:Achar0}
\end{equation}
The characterisation \eqref{eq:Achar0} is the dual representation of
$\mathcal{A}$:
\begin{equation*}
\mathcal{A} = \left\{c\in L^{0}_{+}(\mu):
\langle c,Y^{\gamma}\rangle\leq 1, \quad \mbox{for each
$Y^{\gamma}\in\widetilde{\mathcal{Y}}$}\right\}.
\end{equation*}
This says that $\mathcal{A}$ is the polar of
$\widetilde{\mathcal{Y}}$, establishing \eqref{eq:AY0} and thus
\eqref{eq:AY00}.

Part (i) of the bipolar theorem, Theorem \ref{thm:bp}, along with
\eqref{eq:AY0}, imply that $\mathcal{A}$ is a closed, convex and solid
subset of $L^{0}_{+}(\mu)$ (since it is equal to the polar of a set)
as claimed. Part (ii) of Theorem \ref{thm:bp} gives
$\mathcal{A}^{\circ\circ}\supseteq\mathcal{A}$ with
$\mathcal{A}^{\circ\circ}$ the smallest closed, convex, solid set
containing $\mathcal{A}$. But since $\mathcal{A}$ is itself closed,
convex and solid, we have \eqref{eq:Abipolar}.

\end{proof}

\begin{remark}
\label{rem:Aprop}

There are other ways to obtain the closed, convex and solid properties
of $\mathcal{A}$. First, the equivalence \eqref{eq:Achar0} along with
Fatou's lemma yields that the set $\mathcal{A}$ is closed with respect
to the topology of convergence in measure $\mu$. To see this, let
$(c^{n})_{n\in\mathbb{N}}$ be a sequence in $\mathcal{A}$ which
converges $\mu$-a.e. to an element $c\in L^{0}_{+}(\mu)$. For
arbitrary $Y^{\gamma}\in\widetilde{\mathcal{Y}}$ we obtain, via
Fatou's lemma and the fact that $c^{n}\in\mathcal{A}$ for each
$n\in\mathbb{N}$,
\begin{equation*}
\int_{\mathbf{\Omega}}cY^{\gamma}\ud\mu \leq
\liminf_{n\to\infty}\int_{\mathbf{\Omega}}c^{n}Y^{\gamma}\ud\mu \leq 1,
\end{equation*}
so by \eqref{eq:Achar0}, $c\in\mathcal{A}$, and thus $\mathcal{A}$ is
closed. Further, it is straightforward to establish the convexity of
$\mathcal{A}$ from its definition. Finally, solidity of $\mathcal{A}$
is also clear: if one can dominate a consumption plan
$c\in\mathcal{A}$ with a self-financing wealth process, then one can
also dominate any smaller consumption plan with the same portfolio.
  
\end{remark}

The next step is to attempt to reach some form of reverse polarity
result to \eqref{eq:AY0}. It is here that the enlargement of the dual
domain from $\widetilde{\mathcal{Y}}$ to the set $\mathcal{D}$ of
\eqref{eq:Dy} comes into play.

To see why this enlargement is needed, we first observe from
\eqref{eq:forwardY} that we have
\begin{equation}
Y^{\gamma} \in \widetilde{\mathcal{Y}} \implies
\langle c,Y^{\gamma}\rangle \leq 1, \quad \forall \,
c\in\mathcal{A},
\label{eq:YA}
\end{equation}
which implies that
\begin{equation}
\widetilde{\mathcal{Y}} \subseteq \mathcal{A}^{\circ}.
\label{eq:YA0}  
\end{equation}
We do not have the reverse inclusion, because we do not have the
reverse implication to \eqref{eq:YA}, so cannot write a full
bipolarity relation between sets $\mathcal{A}$ and
$\widetilde{\mathcal{Y}}$. The enlargement from
$\widetilde{\mathcal{Y}}$ to the set $\mathcal{D}$ resolves the issue,
yielding the consumption bipolarity of Lemma \ref{lem:cbp} below. This
procedure, in the spirit of Kramkov and Schachermayer \cite{ks99},
requires us to establish that the enlarged domain is closed in an
appropriate topology. Here is the relevant result.

\begin{lemma}
\label{lem:Dclosed}

The enlarged dual domain $\mathcal{D}\equiv\mathcal{D}(1)$ of
\eqref{eq:Dy} is closed with respect to the topology of convergence in
measure $\mu$.
  
\end{lemma}

The proof of Lemma \ref{lem:Dclosed} will be given further
below. First, we use the result of the lemma to establish the
consumption bipolarity result below.

\begin{lemma}[Consumption bipolarity]
\label{lem:cbp}

Given Lemma \ref{lem:Dclosed}, the set $\mathcal{D}$ is a closed,
convex and solid subset of $L^{0}_{+}(\mu)$, and the the sets
$\mathcal{A}$ and $\mathcal{D}$ satisfy the bipolarity relations
\begin{equation}
\mathcal{A} = \mathcal{D}^{\circ}, \quad \mathcal{D} =
\mathcal{A}^{\circ}.
\label{eq:ADc0}
\end{equation}

\end{lemma}

\begin{proof}

For any $h\in\mathcal{D}$ there will exist an element
$Y^{\gamma}\in\widetilde{\mathcal{Y}}$ such that
$h\leq Y^{\gamma},\,\mu$-almost everywhere. Hence, the implication
\eqref{eq:YA} holds true with $\mathcal{D}$ in place of
$\widetilde{\mathcal{Y}}$:
\begin{equation*}
h\in\mathcal{D} \implies
\langle c,h\rangle\leq 1, \quad \forall\,c\in\mathcal{A},
\end{equation*}
which yields the analogue of \eqref{eq:YA0}:
\begin{equation}
\mathcal{D} \subseteq \mathcal{A}^{\circ}.
\label{eq:DA0}  
\end{equation}
Combining \eqref{eq:AY00} and \eqref{eq:DA0} we have
\begin{equation}
\mathcal{D} \subseteq \widetilde{\mathcal{Y}}^{\circ\circ}.
\label{eq:DY001}  
\end{equation}

Part (ii) of the bipolar theorem, Theorem \ref{thm:bp}, says that
$\widetilde{\mathcal{Y}}^{\circ\circ} \supseteq
\widetilde{\mathcal{Y}}$ and that
$\widetilde{\mathcal{Y}}^{\circ\circ}$ is the smallest closed, convex,
solid set which contains $\widetilde{\mathcal{Y}}$. But $\mathcal{D}$
is also closed, convex and solid (closed due to Lemma
\ref{lem:Dclosed}, convexity following easily from the convexity of
$\widetilde{\mathcal{Y}}$, and solidity is obvious), and by definition
$\mathcal{D}\supseteq\widetilde{\mathcal{Y}}$, so we also have
\begin{equation}
\mathcal{D} \supseteq \widetilde{\mathcal{Y}}^{\circ\circ}.  
\label{eq:DY002}
\end{equation}
Thus, \eqref{eq:DY001} and \eqref{eq:DY002} give
\begin{equation}
\mathcal{D} = \widetilde{\mathcal{Y}}^{\circ\circ}.
\label{eq:DcY00}
\end{equation}
In other words, in enlarging from $\widetilde{\mathcal{Y}}$ to
$\mathcal{D}$ we have succeeded in reaching the bipolar of the former.

Combining \eqref{eq:DcY00} and \eqref{eq:AY00} we see that
$\mathcal{D}$ is the polar of $\mathcal{A}$,
\begin{equation}
\mathcal{D} = \mathcal{A}^{\circ}, 
\label{eq:DcA0}
\end{equation}
so we have the second equality in \eqref{eq:ADc0}. From
\eqref{eq:DcA0} we get $\mathcal{D}^{\circ}=\mathcal{A}^{\circ\circ}$
which, combined with \eqref{eq:Abipolar}, yields the first equality in
\eqref{eq:ADc0}, and the proof is complete.

\end{proof}

It remains to prove Lemma \ref{lem:Dclosed}, which we used above. We
recall the concept of Fatou convergence of stochastic processes from
F\"ollmer and Kramkov \cite{fk97}, that will be needed.

\begin{definition}[Fatou convergence]
\label{def:fc}

Let $(Y^{n})_{n\in\mathbb{N}}$ be a sequence of processes on a
stochastic basis
$(\Omega,\mathcal{F}, \mathbb{F}:=(\mathcal{F}_{t})_{t\geq
  0},\mathbb{P})$, uniformly bounded from below, and let $\tau$ be a
dense subset of $\mathbb{R}_{+}$. The sequence
$(Y^{n})_{n\in\mathbb{N}}$ is said to be \emph{Fatou convergent on
  $\tau$} to a process $Y$ if
\begin{equation*}
Y_{t} = \limsup_{s\downarrow t,\,s\in\tau}\limsup_{n\to\infty}Y^{n}_{s}
=   \liminf_{s\downarrow t,\,s\in\tau}\liminf_{n\to\infty}Y^{n}_{s},
\quad \mbox{a.s $\forall\,t\geq 0$}.
\end{equation*}
If $\tau=\mathbb{R}_{+}$, the sequence is simply called \emph{Fatou
  convergent}. 
  
\end{definition}

The relevant consequence for our purposes is F\"ollmer and Kramkov
\cite[Lemma 5.2]{fk97}, giving a Fatou convergence result for
supermartingales, on a \textit{countable} dense subset of
$\mathbb{R}_{+}$. For the convenience of the reader, we state the
result here.

\begin{lemma}[Fatou convergence of supermartingales, F\"ollmer and
  Kramkov \cite{fk97}, Lemma 5.2]
\label{lem:fk97}

Let $(S^{n})_{n\in\mathbb{N}}$ be a sequence of supermartingales,
uniformly bounded from below, with $S^{n}_{0}=0,\,n\in\mathbb{N}$. Let
$\tau$ be a dense countable subset of $\mathbb{R}_{+}$. Then there is
a sequence $(Y^{n})_{n\in\mathbb{N}}$ of supermartingales, with
$Y^{n}\in\conv(S^{n},S^{n+1},\ldots)$, and a supermartingale $Y$ with
$Y_{0}\leq 0$, such that $(Y^{n})_{n\in\mathbb{N}}$ is Fatou
convergent on $\tau$ to $Y$.
  
\end{lemma}

In Lemma \ref{lem:fk97}, $\conv(S^{n},S^{n+1},\ldots)$ denotes a
convex combination $\sum_{k=n}^{N(n)}\lambda_{k}S^{k}$ for
$\lambda_{k}\in[0,1]$ with $\sum_{k=n}^{N(n)}\lambda_{k}=1$. The
requirement that $S^{n}_{0}=0$ is of course no restriction, since for
a supermartingale with (say) $S^{n}_{0}=1$ (as we shall have when we
apply these results below for supermartingales in $\mathcal{Y}$), we
can always subtract the initial value $1$ to reach a process which
starts at zero.


With this preparation, we can now prove Lemma \ref{lem:Dclosed}.
  
\begin{proof}[Proof of Lemma \ref{lem:Dclosed}]
  
Let $(h^{n})_{n\in\mathbb{N}}$ be a sequence in $\mathcal{D}$,
converging $\mu$-a.e. to some $h\in L^{0}_{+}(\mu)$. We want to show
that $h\in\mathcal{D}$.

Since $h^{n}\in\mathcal{D}$, for each $n\in\mathbb{N}$ we have
$h^{n}\leq\gamma\widehat{Y}^{n},\,\mu$-a.e for some supermartingale
$\widehat{Y}^{n}\in\mathcal{Y}$. With $\tau$ a dense countable subset
of $\mathbb{R}_{+}$, Lemma \ref{lem:fk97} implies that there exists a
sequence $(Y^{n})_{n\in\mathbb{N}}$ of supermartingales, with each
$Y^{n}\in\conv(\widehat{Y}^{n},\widehat{Y}^{n+1},\ldots)$, and a
supermartingale $Y$, such that $(Y^{n})_{n\in\mathbb{N}}$ is Fatou
convergent on $\tau$ to $Y$.

Note that, because $\mathcal{Y}$ is a convex set, for each
$n\in\mathbb{N}$ we have $Y^{n}\in\mathcal{Y}$. Furthermore, by
\v{Z}itkovi\'{c} \cite[Lemma 8]{z02} (proven there for finite horizon
processes, but it is straightforward to verify that the proof goes
through without alteration for infinite horizon processes), there is a
countable set $K\subset\mathbb{R}_{+}$ such that for
$t\in\mathbb{R}_{+}\setminus K$, we have
$Y_{t}=\liminf_{n\to\infty}Y^{n}_{t}$ almost surely, and hence also
$Y=\liminf_{n\to\infty}Y^{n}$, $\mu$-almost everywhere (since these
differ only on a set of measure zero), and indeed
$Y=\liminf_{n\to\infty}Y^{n}$, $\Leb\times\mathbb{P}$-almost
everywhere, where $\Leb$ denotes Lebesgue measure on
$\mathbb{R}_{+}$. We shall use these latter properties shortly.

With $c\in\mathcal{A}$ an admissible consumption plan and $X$ the
associated wealth process, define a supermartingale sequence
$(\widehat{V}^{n})_{n\in\mathbb{N}}$ by
$\widehat{V}^{n}:=X\widehat{Y}^{n} +
\int_{0}^{\cdot}c_{s}\widehat{Y}^{n}_{s}\ud s$. Then, with
$Y^{n}=\sum_{k=n}^{N(n)}\lambda_{k}\widehat{Y}^{k}$ denoting the
convex combination which constructs $(Y^{n})_{n\in\mathbb{N}}$ from
$(\widehat{Y}^{n})_{n\in\mathbb{N}}$, define a corresponding sequence
$(V^{n})_{n\in\mathbb{N}}$ by
\begin{equation*}
V^{n} := \sum_{k=n}^{N(n)}\lambda_{k}\widehat{V}^{k} =
\sum_{k=n}^{N(n)}\lambda_{k}
\left(X\widehat{Y}^{k}+\int_{0}^{\cdot}c_{s}\widehat{Y}^{k}_{s}\ud
 s\right) = XY^{n} + \int_{0}^{\cdot}c_{s}Y^{n}_{s}\ud s.
\end{equation*}
Because $Y^{n}\in\mathcal{Y}$ and $c\in\mathcal{A}$ is an admissible
consumption plan (equivalently, $(X,c)$ is an admissible
investment-consumption strategy),
$XY^{n}+\int_{0}^{\cdot}c_{s}Y^{n}_{s}\ud s$ is a supermartingale for
each $n\in\mathbb{N}$, so that
\begin{equation*}
\mathbb{E}\left[\left.X_{t}Y^{n}_{t} + \int_{s}^{t}c_{u}Y^{n}_{u}\ud
u\right\vert\mathcal{F}_{s}\right] \leq X_{s}Y^{n}_{s}, \quad 0\leq
s\leq t <\infty, \quad n\in\mathbb{N}.
\end{equation*}
Using this, along with the property that $Y=\liminf_{n\to\infty}Y^{n}$,
$\Leb\times\mathbb{P}$-almost everywhere and Fatou's lemma, we have
\begin{eqnarray*}
\mathbb{E}\left[\left.X_{t}Y_{t} + \int_{s}^{t}c_{u}Y_{u}\ud
u\right\vert\mathcal{F}_{s}\right]   & = &
\mathbb{E}\left[\left.\liminf_{n\to\infty}X_{t}Y^{n}_{t} +
\int_{s}^{t}\liminf_{n\to\infty}c_{u}Y^{n}_{u}\ud
u\right\vert\mathcal{F}_{s}\right] \\
& \leq & \liminf_{n\to\infty}\mathbb{E}\left[\left.X_{t}Y^{n}_{t} +
\int_{s}^{t}c_{u}Y^{n}_{u}\ud u\right\vert\mathcal{F}_{s}\right] \\
& \leq &  \liminf_{n\to\infty}X_{s}Y^{n}_{s} \\
& = & X_{s}Y_{s}, \quad  0\leq s\leq t <\infty,
\end{eqnarray*}
which yields the supermartingale property for the process
$XY+\int_{0}^{\cdot}c_{u}Y_{u}\ud u$, and hence that $Y\in\mathcal{Y}$
(since $(X,c)$ is an admissible investment-consumption strategy).

Because $h^{n}\leq\gamma\widehat{Y}^{n},\,\mu$-a.e.for each
$n\in\mathbb{N}$, and since
$Y^{n}=\sum_{k=n}^{N(n)}\lambda_{k}\widehat{Y}^{k}$, we have
\begin{equation}
\gamma Y^{n} =
\sum_{k=n}^{N(n)}\lambda_{k}\gamma\widehat{Y}^{k}
\geq \sum_{k=n}^{N(n)}\lambda_{k}h^{k}, \quad \mbox{$\mu$-a.e.} 
\label{eq:conv}
\end{equation}
Now using that $Y=\liminf_{n\to\infty}Y^{n}$, $\mu$-almost everywhere,
taking the limit inferior in \eqref{eq:conv} and recalling that
$(h^{n})_{n\in\mathbb{N}}$ converges $\mu$-a.e. to $h$, we obtain
\begin{equation*}
\gamma Y \geq \liminf_{n\to\infty}\sum_{k=n}^{N(n)}\lambda_{k}h^{k} = h,
\quad \mbox{$\mu$-a.e.}   
\label{eq:conv1}
\end{equation*}
That is, $h\leq \gamma Y$ $\mu$-a.e, for $Y\in\mathcal{Y}$, so
$h\in\mathcal{D}$, and thus $\mathcal{D}$ is closed.

\end{proof}

With the consumption bipolarity of Lemma \ref{lem:cbp}, we have in
fact established Proposition \ref{prop:abp}, so let us confirm this.

\begin{proof}[Proof of Proposition \ref{prop:abp}]

With the identification $\mathcal{C}=\mathcal{A}$ (from the definition
\eqref{eq:Cx}), and the properties
of $\mathcal{A}$ established in Lemma \ref{lem:Aprop}, we have all the
claimed properties of $\mathcal{C}$ in items (i) and (ii). The
corresponding assertions for $\mathcal{D}$ follow from Lemma
\ref{lem:cbp}.

The property that $\mathcal{C}$ is bounded in $L^{1}(\mu)$ follows
from the second property in \eqref{eq:kappa}, that the integral of
cumulative consumption with respect to
$\kappa\times\mathbb{P}\equiv\mu$, is bounded, and thus
$\mathcal{C}$ is bounded in $L^{1}(\mu)$ and hence in $L^{0}(\mu)$.

For the $L^{0}$-boundedness of $\mathcal{D}$, we shall find a positive
element $\overline{g}\in\mathcal{C}$ and show that $\mathcal{D}$ is
bounded in $L^{1}(\overline{g}\ud\mu)$, and hence bounded in
$L^{0}(\mu)$. Choose
$\mathcal{A}\owns c_{t}\equiv \overline{c}_{t}:=\e^{-\delta t},\,t>0$,
for some $\delta>1$. It is easy to verify that with $x=1$ and
$H\equiv 0$ in \eqref{eq:wealth}, we have $X\geq 0,\,\mu$-a.e., so
$\overline{c}\in\mathcal{A}$. We observe that
$\mathcal{C}\owns\overline{g}\equiv\overline{c}$ is strictly positive
except on a set of $\mu$-measure zero. We then have, for any
$h\in\mathcal{D}$, so $h\leq\gamma Y$ for some $Y\in\mathcal{Y}$
(satisfying $\mathbb{E}[Y_{t}]\leq 1,\,t\geq 0$),
\begin{equation*}
\int_{\mathbf{\Omega}}\overline{g}h\ud\mu \leq
\mathbb{E}\left[\int_{0}^{\infty}\e^{-\delta t}Y_{t}\ud t\right]
= \int_{0}^{\infty}\e^{-\delta t}\mathbb{E}[Y_{t}]\ud t \leq
\frac{1}{\delta}.
\end{equation*}
Thus, $\mathcal{D}$ is bounded in $L^{1}(\overline{g}\ud\mu)$ and
hence bounded in $L^{0}(\mu)$.

\end{proof}

The $L^{1}(\mu)$-boundedness of the primal domain $\mathcal{C}$ is to
be contrasted with the terminal wealth problem of Kramkov and
Schachermayer \cite{ks99,ks03}, in which the dual domain is bounded in
$L^{1}(\mathbb{P})$. This is the source of a switching of roles of the
primal and dual domains in the consumption problem compared with the
terminal wealth problem, and will manifest itself on numerous
occasions in the course of proving the duality theorem in the next
section.

\subsection{Local martingale deflators versus consumption
deflators}
\label{subsec:comp}

We can now return to the discussion of Section \ref{subsec:add}, in
which we made comparisons with the approaches to bipolarity in Chau et
al \cite{chauetal17} and Mostovyi \cite{most15}. This will lead us to
the proof of Proposition \ref{prop:dense}. The proof will demonstrate
that the approach in \cite{most15,chauetal17} can get a fair way
towards establishing bipolarity between $\mathcal{A}$ and the set
$\widetilde{\mathcal{Z}}$, defined analogously to
$\widetilde{\mathcal{Y}}$ in \eqref{eq:Ycbar}, by 
\begin{equation*}
\widetilde{\mathcal{Z}} := \left\{Z^{\gamma}: Z^{\gamma}:=\gamma Z, \,
  Z\in\mathcal{Z}\right\}. 
\end{equation*}
One can get a little further by enlarging to $D$, but it is then
necessary to invoke the closure $\overline{D}$ to reach full
bipolarity. This establishes the result of the proposition, and shows
how the dual domain we chose is not too big, and not too small, to
establish bipolarity. A byproduct of the proof is that it shows how
results analogous to Mostovyi \cite[Lemma 4.2]{most15} and Chau et al
\cite[Lemma 1]{chauetal17}, which give an equivalence between an
admissible consumption plan and an appropriate budget constraint
involving either local martingale deflators (in \cite{chauetal17}) or
martingale deflators (in \cite{most15}) can be established without
recourse to constructions involving equivalent measures, by judicious
use of the Stricker and Yan \cite{sy98} ODT, rather like the proof of
Lemma \ref{lem:suffc}. As we pointed out in Section
\ref{subsubsec:completion}, this is both an aesthetic and
mathematically desirable feature.

\begin{proof}[Proof of Proposition \ref{prop:dense}]

Consider a consumption plan with initial capital $x=1$. Using the same
arguments as in Remark \ref{rem:alternatives} we establish the
analogue of \eqref{eq:EcZ} for $x=1$:
\begin{equation}
\mathbb{E}\left[\int_{0}^{t}c_{s}Z_{s}\ud s\right] \leq 1 -
\mathbb{E}\left[\int_{0}^{t}C_{s-}\ud Z_{s}\right], \quad
t\geq 0.
\label{eq:EcZ1}  
\end{equation}
The process $M:=\int_{0}^{\cdot}C_{s-}\ud Z_{s}$ is a local
martingale. With $(T_{n})_{n\in\mathbb{N}}$ a localising sequence for
$M$, so that
$\mathbb{E}\left[\int_{0}^{T_{n}}C_{s-}\ud
  Z_{s}\right]=0,\,n\in\mathbb{N}$, \eqref{eq:EcZ1} converts to
$\mathbb{E}\left[\int_{0}^{T_{n}}c_{s}Z_{s}\ud s\right] \leq
1,\,n\in\mathbb{N}$. Letting $n\uparrow\infty$ and using monotone
convergence we obtain a budget constraint in the form
$\mathbb{E}\left[\int_{0}^{\infty}c_{t}Z_{t}\ud t\right]\leq 1$. We
thus have the implications analogous to \eqref{eq:forwardc} and
\eqref{eq:forwardY}:
\begin{equation}
c\in\mathcal{A} \implies
\mathbb{E}\left[\int_{0}^{\infty}c_{t}Z_{t}\ud t\right]\leq 1, \quad
\forall\,Z\in\mathcal{Z},
\label{eq:forwardcZ}
\end{equation}
and
\begin{equation}
Z\in\mathcal{Z} \implies
\mathbb{E}\left[\int_{0}^{\infty}c_{t}Z_{t}\ud t\right]\leq 1, \quad
\forall\,c\in\mathcal{A}.
\label{eq:forwardZ}
\end{equation}
We can then establish the reverse implication to \eqref{eq:forwardcZ}
in exactly the same manner as in the proof of Lemma
\ref{lem:suffc}. In other words, if $c$ is a non-negative process
satisfying the budget constraint, then it is an admissible consumption
plan. That is, we have
\begin{equation}
\mathbb{E}\left[\int_{0}^{\infty}c_{t}Z_{t}\ud t\right]\leq 1, \,
\forall\, Z\in\mathcal{Z} \implies c\in\mathcal{A}.
\label{eq:ecz1}
\end{equation}
Thus, following the same arguments as in the proof of Lemma
\ref{lem:Aprop}, we have, from \eqref{eq:forwardcZ} and
\eqref{eq:ecz1},
\begin{equation*}
\mathcal{A} = \left\{c\in L^{0}_{+}(\mu):
\langle c,\gamma Z\rangle\leq 1, \quad \mbox{for each
$Z^{\gamma}=\gamma Z\in\widetilde{\mathcal{Z}}$}\right\},
\end{equation*}
so that $\mathcal{A}$ is the polar of $\widetilde{\mathcal{Z}}$:
\begin{equation}
\mathcal{A} = \widetilde{\mathcal{Z}}^{\circ},
\label{eq:AZ0}
\end{equation}
implying
\begin{equation}
\mathcal{A}^{\circ} = \widetilde{\mathcal{Z}}^{\circ\circ},
\label{eq:AZ00}
\end{equation}
and that $\mathcal{A}$ is equal to its bipolar:
\begin{equation*}
\mathcal{A}^{\circ\circ} = \mathcal{A},
\end{equation*}
by the same arguments as in the proof of Lemma \ref{lem:Aprop}.

Now, \eqref{eq:forwardZ} gives us that
\begin{equation*}
\widetilde{\mathcal{Z}} \subseteq \mathcal{A}^{\circ},
\end{equation*}
by the same arguments that led to \eqref{eq:YA0}. We do not have the
reverse inclusion, because we do not have the reverse implication to
\eqref{eq:forwardZ}, so cannot write a full bipolarity relation
between sets $\mathcal{A}$ and $\widetilde{\mathcal{Z}}$. To this end,
one can try enlarging the dual domain from $\widetilde{\mathcal{Z}}$
to $D$, in the same manner that we enlarged from
$\widetilde{\mathcal{Y}}$ to the set $\mathcal{D}$ when using
consumption deflators as dual variables. This yields, in the same
manner as we established \eqref{eq:DA0},
\begin{equation}
D \subseteq \mathcal{A}^{\circ}.
\label{eq:DA01}  
\end{equation}
Combining \eqref{eq:AZ00} and \eqref{eq:DA01} we have
\begin{equation}
D \subseteq \widetilde{\mathcal{Z}}^{\circ\circ}.
\label{eq:DZ001}  
\end{equation}
Here is the crucial point: to establish the reverse inclusion to
\eqref{eq:DZ001} would require that the set $D$ is closed with respect
to the topology of convergence in measure $\mu$. But the arguments we
used for the proof of Lemma \ref{lem:Dclosed} to establish this
property for the domain $\mathcal{D}$, break down when applied to the
set $D$, because the limiting supermartingale in the Fatou convergence
argument is known only to be a supermartingale, and cannot be shown to
be a local martingale deflator. So we are forced to enlarge $D$ itself
to its closure $\overline{D}$.

With this enlargement to $\overline{D}$, we first show that
\eqref{eq:DA01}, and hence \eqref{eq:DZ001}, extend from $D$ to
$\overline{D}$. Suppose $(h^{n})_{n\in\mathbb{N}}$ is a sequence in
$D\subseteq\overline{D}$ that converges $\mu$-a.e. to some element
$h\in L^{0}_{+}(\mu)$. Then $h\in\overline{D}$, since $\overline{D}$
is closed in $\mu$-measure (and a $\mu$-a.e. convergent sequence
must also converge in measure $\mu$). Using Fatou's lemma and that
$h^{n}\in D$ we have, for any $c\in\mathcal{A}$
\begin{equation*}
\langle c,h\rangle = \langle c,\lim_{n\to\infty}h^{n}\rangle \leq
\lim_{n\to\infty}\langle c,h^{n}\rangle \leq 1.  
\end{equation*}
Thus, we get the implication
$h\in\overline{D}\implies \langle c,h\rangle\leq 1,\,\forall
c\in\mathcal{A}$, so we extend \eqref{eq:DA01} and, in particular,
\eqref{eq:DZ001} from $D$ to $\overline{D}$:
\begin{equation*}
\overline{D} \subseteq \widetilde{\mathcal{Z}}^{\circ\circ},
\end{equation*}
which is the analogue of \eqref{eq:DY001}. Finally, using the bipolar
theorem in the same manner as the last part of the proof of Lemma
\ref{lem:cbp}, we establish bipolarity between $\overline{D}$ and
$\mathcal{A}$:
\begin{equation}
\mathcal{A} = \overline{D}^{\circ}, \quad \overline{D} =
\mathcal{A}^{\circ}.
\label{eq:AoDc0}
\end{equation}
Comparing \eqref{eq:AoDc0} with \eqref{eq:ADc0} shows that we have
$\mathcal{D}=\overline{D}$, so $D$ is dense in $\mathcal{D}$, and the
proof is complete.

\end{proof}

\begin{remark}[Relations between the bipolars of dual domains]
\label{rem:relbp}

We can now round off the discussion initiated in Remark
\ref{rem:alternatives}, regarding relations between the various dual
domains that one might use in establishing a consumption duality.  In
particular, we examine how these relations transform when passing to
the bipolar in the product space.

For brevity, in this remark we shall use the product space
$\nu:=\Leb\times\mathbb{P}$. Similar remarks pertain with respect to
$\mu=\kappa\times\mathbb{P}$, bringing in the extraneous factor
$\gamma$ in defining the solid hulls of the dual domains. Denote by
$\polar(A)$ the polar of any set $A\subset L^{0}_{+}(\nu)$, with
$\bipolar(\cdot)$ and $\solid(\cdot)$ denoting the bipolar and solid
hull, respectively.

In the original dual space, we have the inclusions
\begin{equation*}
\mathcal{Z} \subseteq \mathcal{Y} \subseteq \mathcal{Y}^{0},
\end{equation*}
as noted in Section \ref{subsec:dfcp}. The budget constraint of Lemma
\ref{lem:ihbc} combined with Lemma \ref{lem:suffc} yielded the
properties in Lemma \ref{lem:Aprop}, and in particular the property
\eqref{eq:AY0} that the primal domain was the polar of the original
dual domain of consumption deflators. The proof of Proposition
\ref{prop:dense} showed that the same property held with local
martingale deflators as dual variables (see \eqref{eq:AZ0}). Thus, in
the notation of this remark, we have
\begin{equation*}
\polar(\mathcal{Z}) = \polar(\mathcal{Y}), 
\end{equation*}
with both sets equal to the primal domain $\mathcal{A}$.

Passing to the bipolar, we established perfect bipolarity between the
primal and dual domains by enlarging the dual space to its solid hull
and showing (in Lemma \ref{lem:Dclosed}) that the resulting domain was
closed. In the notation of this remark, we have
$\bipolar(\mathcal{Y})=\solid(\mathcal{Y})$, and thus
\begin{equation}
\bipolar(\mathcal{Z}) = \solid(\mathcal{Y}).   
\label{eq:bipolar1}
\end{equation}
The content of Proposition \ref{prop:dense} is that we cannot replace
$\bipolar(\mathcal{Z})$ with $\solid(\mathcal{Z})$ in
\eqref{eq:bipolar1} unless we take the closure,
$\cl(\solid(\mathcal{Z}))$, becuase the Fatou supermartingale
convergence method could not guarantee a local martingale deflator as
the limiting supermartingale.

The final step in this chain of results is to incorporate the set
$\mathcal{Y}^{0}$ of wealth deflators. Clearly
$\solid(\mathcal{Y}^{0})\supseteq\solid(\mathcal{Y})$. Moreover, by a
similar (and easier) proof as for Lemma \ref{lem:Dclosed},
$\solid(\mathcal{Y}^{0})$ is closed (in essence, a convex combination
of wealth deflators Fatou converges to a supermartingale, while the
corresponding deflated self-financing wealth also converges to a
deflated wealth supermartingale, so the limiting supermartingale is a
wealth deflator, and the rest of the proof rests on the solidity of
$\solid(\mathcal{Y}^{0})$). The remaining arguments are as for
$\solid(\mathcal{Y})$, except for the crucial proviso that, as
indicated in Remark \ref{rem:alternatives}, we are not able to
establish the infinite horizon budget constraint with wealth
deflators, so we appear to have
$\mathcal{A}\not\subseteq\polar(\mathcal{Y}^{0})$, indicating that the
inclusion $\solid(\mathcal{Y}^{0})\supset\solid(\mathcal{Y})$ is
strict. We thus have
\begin{equation*}
\bipolar(\mathcal{Z}) = \solid(\mathcal{Y})
\subset\solid(\mathcal{Y}^{0}). 
\end{equation*}

This remark thus indicates some open questions. If one could show that
the budget constrant holds with wealth deflators, then one would
arrive at the equality $\mathcal{Y}=\mathcal{Y}^{0}$. (By the same
token, if one were able to show that $\solid(\mathcal{Z})$ is closed,
one would arrive at $\mathcal{Z}=\mathcal{Y}$.) One can envisage
concrete models where the local martingale deflators and the wealth
deflators coincide with the consumption deflators (think of the
Brownian models that form the building block of the monograph of
Karatzas and Shreve \cite{ks98}). The message here is that, at the
level of abstraction of this paper, utilising Fatou convergence
techniques for supermartingales, it does not appear possible to show
that the dual domains coalesce. It would appear that in order to
establish (for instance) $\mathcal{Y}=\mathcal{Y}^{0}$ in full
generality, would require new techniques.
  
\end{remark}

\section{Proofs of the duality theorems}
\label{sec:pdt}

In this section we prove the abstract duality of Theorem
\ref{thm:adt}, from which the concrete duality of Theorem
\ref{thm:consd} is then deduced. Throughout this section, we have in
place the result of Proposition \ref{prop:abp}, as this bipolarity is
the starting point of the duality proof. The proof of Theorem
\ref{thm:adt} proceeds via a series of lemmas. Some of them have a
similar flavour to the steps in the celebrated Kramkov and
Schachermayer \cite{ks99,ks03} abstract duality proof, but in many
places the roles of the primal and dual domains are reversed compared
to \cite{ks99,ks03}. This is because in \cite{ks99,ks03} the dual
domain is $L^{1}(\mathbb{P})$-bounded, but here it is the primal
domain that is $L^{1}(\mu)$-bounded.

Let us state the basic properties that are taken as given throughout
this section.

\begin{fact}
\label{fact:bf}
  
Throughout this section, assume that the utility function satisfies
the Inada conditions \eqref{eq:inada}, that the sets $\mathcal{C}$ and
$\mathcal{D}$ satisfy all the properties in Proposition
\ref{prop:abp}, and that the abstract primal and dual value functions
in \eqref{eq:vfabs} and \eqref{eq:dvfabs} satisfy the minimal
conditions in \eqref{eq:minimal}.

\end{fact}

\emph{All subsequent lemmata and propositions in this section implicitly
take Fact \ref{fact:bf} as given.}

The first step is to establish weak duality. 

\begin{lemma}[Weak duality]
\label{lem:weakdual}
  
The primal and dual value functions $u(\cdot)$ and $v(\cdot)$ of
\eqref{eq:vfabs} and \eqref{eq:dvfabs} satisfy the weak duality bounds
\begin{equation}
v(y) \geq \sup_{x>0}[u(x)-xy], \quad y>0, \quad
\mbox{equivalently} \quad u(x) \leq \inf_{y>0}[v(y)+xy], \quad
x>0. 
\label{eq:weakd}
\end{equation}
As a result, $u(x)$ is finitely valued for all $x>0$. Moreover, we
have the limiting relations
\begin{equation}
\limsup_{x\to\infty}\frac{u(x)}{x} \leq 0, \quad
\liminf_{y\to\infty}\frac{v(y)}{y} \geq 0. 
\label{eq:limiting}
\end{equation}

\end{lemma}

\begin{proof}

For any $g\in\mathcal{C}(x)$ and $h\in\mathcal{D}(y)$, using the
polarity relations in  \eqref{eq:bp1} and \eqref{eq:bp2} we may bound
the achievable utility according to
\begin{eqnarray}
\int_{\mathbf{\Omega}}U(g)\ud\mu & \leq &
\int_{\mathbf{\Omega}}U(g)\ud\mu + xy -
\int_{\mathbf{\Omega}}gh\ud\mu \nonumber \\
& = & \int_{\mathbf{\Omega}}(U(g)-gh)\ud\mu + xy \nonumber \\ 
& \leq & \int_{\mathbf{\Omega}}V(h)\ud\mu + xy, \quad x,y>0, 
\label{eq:ucbound}         
\end{eqnarray}
the last inequality a consequence of \eqref{eq:VUbound}. Maximising
the left-hand-side of \eqref{eq:ucbound} over $g\in\mathcal{C}(x)$ and
minimising the right-hand-side over $h\in\mathcal{D}(y)$ gives
$u(x)\leq v(y)+xy$ for all $x,y>0$, and \eqref{eq:weakd} follows.

The assumption that $v(y)<\infty$ for all $y>0$ immediately yields
that $u(x)$ is finitely valued for some $x>0$. Since $U(\cdot)$ is
strictly increasing and strictly concave, and given the convexity of
$\mathcal{C}$, these properties are inherited by $u(\cdot)$, which is
therefore finitely valued for all $x>0$. Finally, the relations in
\eqref{eq:weakd} easily lead to those in \eqref{eq:limiting}.
  
\end{proof}

Above, we obtained concavity and monotonicity of $u(\cdot)$ by using
convexity of $\mathcal{C}$ and the properties of $U(\cdot)$. Similar
arguments show that $v(\cdot)$ is strictly decreasing and strictly
convex. We shall see these properties reproduced in proofs of
existence and uniqueness of the optimisers for $u(\cdot),v(\cdot)$.

The next step is to give a compactness lemma for the primal
domain.

\begin{lemma}[Compactness lemma for $\mathcal{C}$]
\label{lem:Ccompact}

Let $(\tilde{g}^{n})_{n\in\mathbb{N}}$ be a sequence in
$\mathcal{C}$. Then there exists a sequence $(g^{n})_{n\in\mathbb{N}}$
with $g^{n}\in\conv(\tilde{g}^{n}, \tilde{g}^{n+1},\ldots)$, which
converges $\mu$-a.e. to an element $g\in\mathcal{C}$ that is
$\mu$-a.e. finite.

\end{lemma}

\begin{proof}

Delbaen and Schachermayer \cite[Lemma A1.1]{ds94} (adapted from a
probability space to the finite measure space
$(\mathbf{\Omega},\mathcal{G},\mu)$) implies the existence of a
sequence $(g^{n})_{n\in\mathbb{N}}$, with
$g^{n}\in\conv(\tilde{g}^{n}, \tilde{g}^{n+1},\ldots)$, which
converges $\mu$-a.e. to an element $g$ that is $\mu$-a.e. finite
because $\mathcal{C}$ is bounded in $L^{0}(\mu)$ (the finiteness also
following from \cite[Lemma A1.1]{ds94}). By convexity of
$\mathcal{C}$, each $g^{n},\,n\in\mathbb{N}$ lies in
$\mathcal{C}$. Finally, by Fatou's lemma, for every $h\in\mathcal{D}$
we have
\begin{equation*}
\int_{\mathbf{\Omega}}gh\ud\mu =
\int_{\mathbf{\Omega}}\liminf_{n\to\infty}g^{n}h\ud\mu \leq
\liminf_{n\to\infty}\int_{\mathbf{\Omega}}g^{n}h\ud\mu \leq 1,  
\end{equation*}
so that $g\in\mathcal{C}$.

\end{proof}

Results in the style of Lemma \ref{lem:Ccompact} are standard in these
duality proofs. We will see a similar result for the dual domain
$\mathcal{D}$ shortly. Typically, the program is to first prove such a
result in the dual domain and to follow this with a uniform
integrability result for the family
$(V^{-}(h))_{h\in\mathcal{D}(y)}$. This facilitates a proof of
existence and uniqueness of the dual minimiser, and of the conjugacy
for the value functions by establishing the first relation in
\eqref{eq:conjugacy}.

Here, as we have alluded to earlier, the natural course of events is
switched on its head: one works instead first in the primal domain,
with the next step to prove a uniform integrability result for the
family $(U^{+}(g))_{g\in\mathcal{C}(x)}$. This leads to existence and
uniqueness of the primal maximiser, and to conjugacy in the form of
the second (bi-conjugate) relation in \eqref{eq:conjugacy}. The style
of proof in the dual domain for the classical program transfers to the
primal domain here. This switching of the roles of the primal and dual
domains will be an almost continual feature of the analysis of this
section, and we shall point out further instances of it in due
course. All this stems from the $L^{1}(\mu)$-boundedness of the primal
(as opposed to the dual) domain in the consumption problem, as pointed
out in the first paragraph of this section.

Here is the next step in this chain of results.

\begin{lemma}[Uniform integrability of
$(U^{+}(g))_{g\in\mathcal{C}(x)}$]
\label{lem:Uplusg}

The family $(U^{+}(g))_{g\in\mathcal{C}(x)}$ is uniformly integrable,
for any $x>0$.
  
\end{lemma}

The style of the proof is along identical lines to Kramkov and
Schachermayer \cite[Lemma 3.2]{ks99}, but there it was applied to the
concave function $-V(\cdot)$ and in the dual domain to prove the
uniform integrability of $(V^{-}(h))_{h\in\mathcal{D}(y)}$. We are
witnessing the switching of the roles of $\mathcal{C}$ and
$\mathcal{D}$.

\begin{proof}[Proof of Lemma \ref{lem:Uplusg}]

  Since $U(\cdot)$ is increasing, we need only consider the case where
  $U(\infty):=\lim_{x\to\infty}U(x)=+\infty$ (otherwise there is
  nothing to prove). Let $\varphi:(U(0),U(\infty))\mapsto(0,\infty)$
  denote the inverse of $U(\cdot)$. Then $\varphi(\cdot)$ is strictly
  increasing. For any $g\in\mathcal{C}(x)$ (so
  $\int_{\mathbf{\Omega}}g\ud\mu\leq Kx$, for some $K<\infty$) we
  have, for all $x>0$,
\begin{equation*}
\int_{\mathbf{\Omega}}\varphi(U^{+}(g))\ud\mu \leq \varphi(0) +
\int_{\mathbf{\Omega}}\varphi(U(g))\ud\mu = \varphi(0) +
\int_{\mathbf{\Omega}}g\ud\mu \leq \varphi(0) + Kx.
\end{equation*}
Then, using l'H\^opital's rule and the change of variable
$\varphi(x)=y\iff x=U(y)$, we have
\begin{equation}
\lim_{x\to U(\infty)}\frac{\varphi(x)}{x} =
\lim_{x\to\infty}\frac{\varphi(x)}{x} =
\lim_{y\to\infty}\frac{y}{U(y)} =
\lim_{y\to\infty}\frac{1}{U^{\prime}(y)} = +\infty, 
\label{eq:phioverx}
\end{equation}
on using the Inada conditions \eqref{eq:inada}. The
$L^{1}(\mu)$-boundedness of $\mathcal{C}(x)$ means we can apply the de
la Vall\'ee-Poussin theorem (Pham \cite[Theorem A.1.2]{pham09}) which,
combined with \eqref{eq:phioverx}, implies the uniform integrability
of the family $(U^{+}(g))_{g\in\mathcal{C}(x)}$.

\end{proof}

\begin{remark}
\label{rem:Uplusg}

There is another way to establish Lemma \ref{lem:Uplusg} which matches
more closely the style of proof in Kramkov and Schachermayer
\cite[Lemma 1]{ks03}, and which we shall see applied to the dual
domain in Lemma \ref{lem:Vminush} to establish uniform integrability
of $(V^{-}(h))_{h\in\mathcal{D}(y)}$. We mention this method here,
because at first glance the method of \cite[Lemma 1]{ks03} will not
work to establish Lemma \ref{lem:Uplusg}, due to the fact that
$\mathcal{D}$ is not bounded in $L^{1}(\mu)$. However, as we show
here, a slight adjustment to the proof can rectify matters. Here is
the argument.

Let $(g^{n})_{n\in\mathbb{N}}$ be a sequence in $\mathcal{C}(x)$, for
any fixed $x>0$. We want to show that the sequence
$(U^{+}(g^{n}))_{n\in\mathbb{N}}$ is uniformly integrable. 

Fix $x>0$. If $U(\infty)\leq 0$ there is nothing to prove, so assume
$U(\infty)>0$.

If the sequence $(U^{+}(g^{n}))_{n\in\mathbb{N}}$ is not uniformly
integrable, then, passing if need be to a subsequence still denoted by 
$(g^{n})_{n\in\mathbb{N}}$, we can find a constant $\alpha>0$ and a
disjoint sequence $(A_{n})_{n\in\mathbb{N}}$ of sets of
$(\mathbf{\Omega},\mathcal{G})$ (so
$A_{n}\in\mathcal{G},\,n\in\mathbb{N}$ and $A_{i}\cap A_{j}=\emptyset$
if $i\neq j$) such that
\begin{equation*}
\int_{\mathbf{\Omega}}U^{+}(g^{n})\mathbbm{1}_{A_{n}}\ud\mu \geq
\alpha, \quad n\in\mathbb{N}.  
\end{equation*}
(See for example Pham \cite[Corollary A.1.1]{pham09}.) Define, for
some $g^{0}\in\mathcal{C}$, a sequence $(f^{n})_{n\in\mathbb{N}}$ of
elements in $L^{0}_{+}(\mu)$ by
\begin{equation}
f^{n} := x_{0}g^{0} + \sum_{k=1}^{n}g^{k}\mathbbm{1}_{A_{k}},   
\label{eq:fn}
\end{equation}
where $x_{0}:=\inf\{x>0:\,U(x)\geq 0\}$. (It is here where we are
amending the arguments in Kramkov and Schachermayer \cite[Lemma
1]{ks03}: there, one defines the sequence $(f^{n})_{n\in\mathbb{N}}$
by $f^{n}:=x_{0}+\sum_{k=1}^{n}g^{k}\mathbbm{1}_{A_{k}}$, but an
examination of the rest of the argument we now give shows that this
will require
$\int_{\mathbf{\Omega}}h\ud\mu\leq 1,\,\forall h\in\mathcal{D}$, which
we do not have, because the constant consumption stream $c\equiv 1$ is
not admissible. But we do have instead
$\int_{\mathbf{\Omega}}gh\ud\mu\leq 1,\,\forall g\in\mathcal{C},
h\in\mathcal{D}$, which allows the alternative definition of the
sequence $(f^{n})_{n\in\mathbb{N}}$ in \eqref{eq:fn} to make things
work.)

For any $h\in\mathcal{D}$ we have
\begin{equation*}
\int_{\mathbf{\Omega}}f^{n}h\ud\mu =
\int_{\mathbf{\Omega}}\left(x_{0}g^{0} +
\sum_{k=1}^{n}g^{k}\mathbbm{1}_{A_{k}}\right)h\ud\mu \leq x_{0} +
\sum_{k=1}^{n}\int_{\mathbf{\Omega}}g^{k}h\mathbbm{1}_{A_{k}}\ud\mu
\leq x_{0} + nx.
\end{equation*}
Thus, $f^{n}\in\mathcal{C}(x_{0}+nx),\,n\in\mathbb{N}$.

On the other hand, since $U^{+}(\cdot)$ is non-negative and
non-decreasing,
\begin{eqnarray*}
\int_{\mathbf{\Omega}}U(f^{n})\ud\mu & = &
\int_{\mathbf{\Omega}}U^{+}(f^{n})\ud\mu \\
& = & \int_{\mathbf{\Omega}}U^{+}\left(x_{0}g^{0} +
\sum_{k=1}^{n}g^{k}\mathbbm{1}_{A_{k}}\right)\ud\mu \\
& \geq & \int_{\mathbf{\Omega}}
U^{+}\left(\sum_{k=1}^{n}g^{k}\mathbbm{1}_{A_{k}}\right)\ud\mu \\
& = & \sum_{k=1}^{n}\int_{\mathbf{\Omega}}
U^{+}\left(g^{k}\mathbbm{1}_{A_{k}}\right)\ud\mu \geq \alpha n.
\end{eqnarray*}
Therefore,
\begin{equation*}
\limsup_{z\to\infty}\frac{u(z)}{z} =
\limsup_{n\to\infty}\frac{u(x_{0}+nx)}{x_{0}+nx} \geq
\limsup_{n\to\infty}\frac{\int_{\mathbf{\Omega}}U(f^{n})\ud\mu}{x_{0}+nx}
\geq \limsup_{n\to\infty}\left(\frac{\alpha n}{x_{0}+nx}\right) =
\frac{\alpha}{x} > 0,
\end{equation*}
which contradicts the limiting weak duality bound in
\eqref{eq:limiting}. This contradiction establishes the result.

\end{remark}

One can can now proceed to prove either existence of a unique
optimiser in the primal problem, or conjugacy of the value
functions. We proceed first the former, followed by conjugacy.

\begin{lemma}[Primal existence]
\label{lem:primexis}

The optimal solution $\widehat{g}(x)\in\mathcal{C}(x)$ to the primal
problem \eqref{eq:vfabs} exists and is unique, so that $u(\cdot)$ is
strictly concave.  
  
\end{lemma}

\begin{proof}

Fix $x>0$. Let $(g^{n})_{n\in\mathbb{N}}$ be a maximising sequence
in $\mathcal{C}(x)$ for $u(x)<\infty$ (the finiteness proven in
Lemma \ref{lem:weakdual}). That is
\begin{equation}
\lim_{n\to\infty}\int_{\mathbf{\Omega}}U(g^{n})\ud\mu = u(x) < \infty.
\label{eq:maxseqprimal}
\end{equation}

By the compactness lemma for $\mathcal{C}$ (and thus also for
$\mathcal{C}(x)=x\mathcal{C}$), Lemma \ref{lem:Ccompact}, we can find
a sequence $(\widehat{g}^{n})_{n\in\mathbb{N}}$ of convex
combinations, so
$\mathcal{C}(x)\owns\widehat{g}^{n}\in \conv(g^{n},
g^{n+1},\ldots),\,n\in\mathbb{N}$, which converges $\mu$-a.e. to some
element $\widehat{g}(x)\in\mathcal{C}(x)$. We claim that
$\widehat{g}(x)$ is the primal optimiser. That is, that we have
\begin{equation}
\int_{\mathbf{\Omega}}U(\widehat{g}(x))\ud\mu = u(x).
\label{eq:primaloptimiser}
\end{equation}
By concavity of $U(\cdot)$ and \eqref{eq:maxseqprimal} we have
\begin{equation*}
\lim_{n\to\infty}\int_{\mathbf{\Omega}}U(\widehat{g}^{n})\ud\mu \geq
\lim_{n\to\infty}\int_{\mathbf{\Omega}}U(g^{n})\ud\mu = u(x),
\end{equation*}
which, combined with the obvious inequality
$u(x)\geq
\lim_{n\to\infty}\int_{\mathbf{\Omega}}U(\widehat{g}^{n})\ud\mu$ means
that we also have, further to \eqref{eq:maxseqprimal},
\begin{equation*}
\lim_{n\to\infty}\int_{\mathbf{\Omega}}U(\widehat{g}^{n})\ud\mu =
u(x).  
\end{equation*}
In other words
\begin{equation}
\lim_{n\to\infty}\int_{\mathbf{\Omega}}U^{+}(\widehat{g}^{n})\ud\mu -
\lim_{n\to\infty}\int_{\mathbf{\Omega}}U^{-}(\widehat{g}^{n})\ud\mu =
u(x) < \infty,    
\label{eq:Uplusminus}
\end{equation}
and note therefore that both integrals in \eqref{eq:Uplusminus} are
finite.

From Fatou's lemma, we have
\begin{equation}
\lim_{n\to\infty}\int_{\mathbf{\Omega}}U^{-}(\widehat{g}^{n})\ud\mu
\geq \int_{\mathbf{\Omega}}U^{-}(\widehat{g}(x))\ud\mu.  
\label{eq:Uminus}
\end{equation}
From Lemma \ref{lem:Uplusg} we have uniform integrability of
$(U^{+}(\widehat{g}^{n}))_{n\in\mathbb{N}}$, so that
\begin{equation}
\lim_{n\to\infty}\int_{\mathbf{\Omega}}U^{+}(\widehat{g}^{n})\ud\mu
= \int_{\mathbf{\Omega}}U^{+}(\widehat{g}(x))\ud\mu.  
\label{eq:Uplus}
\end{equation}
Thus, using \eqref{eq:Uminus} and \eqref{eq:Uplus} in
\eqref{eq:Uplusminus}, we obtain
\begin{equation*}
u(x) \leq \int_{\mathbf{\Omega}}U(\widehat{g}(x))\ud\mu,
\end{equation*}
which, combined with the obvious inequality
$u(x)\geq\int_{\mathbf{\Omega}}U(\widehat{g}(x))\ud\mu$, yields
\eqref{eq:primaloptimiser}. The uniqueness of the primal optimiser
follows from the strict concavity of $U(\cdot)$, as does the strict
concavity of $u(\cdot)$. For this last claim, fix $x_{1}<x_{2}$ and
$\lambda\in(0,1)$, note that
$\lambda\widehat{g}(x_{1}) +
(1-\lambda)\widehat{g}(x_{2})\in\mathcal{C}(\lambda x_{1} +
(1-\lambda)x_{2})$ (yet must be sub-optimal for
$u(\lambda x_{1}+(1-\lambda)x_{2})$ as it is not guaranteed to equal
$\widehat{g}(\lambda x_{1}+(1-\lambda)x_{2})$) and therefore, using
the strict concavity of $U(\cdot)$,
\begin{equation*}
u(\lambda x_{1} + (1-\lambda)x_{2})  \geq
\int_{\mathbf{\Omega}}U\left(\lambda\widehat{g}(x_{1}) +
  (1-\lambda)\widehat{g}(x_{2})\right)\ud\mu > \lambda u(x_{1}) +
(1-\lambda)u(x_{2}). 
\end{equation*}

\end{proof}

We now establish conjugacy of the value functions. Compared with the
classical method of proof in Kramkov and Schachermayer \cite[Lemma
3.4]{ks99}, our method is similar, but instead of bounding the
elements in the primal domain to create a compact set for the weak$*$
topology $\sigma(L^{\infty},L^{1})$ on $L^{\infty}(\mu)$, we bound the
elements in the dual domain.\footnote{Recall that a sequence
  $(h^{n})_{n\in\mathbb{N}}$ in $L^{\infty}(\mu)$ converges to
  $h\in L^{\infty}(\mu)$ with respect to the weak$*$ topology
  $\sigma(L^{\infty},L^{1})$ if and only if $\langle g,h^{n}\rangle$
  converges to $\langle g,h\rangle$ for each $g\in L^{1}(\mu)$.}
Accordingly, we apply a version of the minimax theorem with a
minimisation over a compact set and a maximisation over a subset of a
vector space (see, for example, Aubin and Ekeland \cite[Theorem 7,
Page 319]{ak84}), as opposed to the maximisation over a compact set
and a minimisation over a subset of a vector space (as in Strasser
\cite[Theorem 45.8]{strasser85}). (See also Sion \cite[Theorem 3.2 and
Corollary 3.3]{sion58}, in which either one of the convex spaces
involved can be taken to be compact.) This reversal is appropriate
because the primal domain is a subset of $L^{1}(\mu)$, whereas in the
terminal wealth problem the dual domain is a subset of
$L^{1}(\mathbb{P})$. The consequence is that we prove the second
(bi-conjugate) relation in \eqref{eq:conjugacy}, as opposed to the
first. Here is the minimax theorem as we shall use it.

\begin{theorem}[Minimax]
\label{thm:rminimax}

Let $\mathcal{X}$ be a convex subset of a normed vector space $E$ and
let $\mathcal{Y}$ be a $\sigma(E^{\prime},E)$-compact convex, subset
of the topological dual $E^{\prime}$ of $E$. Assume that
$f:\mathcal{X}\times\mathcal{Y}\to\mathbb{R}$ satisfies the following
conditions:

\begin{enumerate}

\item $x\mapsto f(x,y)$ is concave on $\mathcal{X}$ for every
$y\in\mathcal{Y}$;

\item $y\mapsto f(x,y)$ is lower semicontinuous and convex on
$\mathcal{Y}$ for every $x\in\mathcal{X}$.
  
\end{enumerate}

Then:
\begin{equation*}
\inf_{y\in\mathcal{Y}}\sup_{x\in\mathcal{X}}f(x,y) =
\sup_{x\in\mathcal{X}}\inf_{y\in\mathcal{Y}}f(x,y). 
\end{equation*}
  
\end{theorem}

Here is the conjugacy result for the primal and dual value functions.

\begin{lemma}[Conjugacy]
\label{lem:conjugacy}
  
The primal value function in \eqref{eq:vfabs} satisfies the
bi-conjugacy relation 
\begin{equation*}
u(x) = \inf_{y>0}[v(y)+xy], \quad \mbox{for each $x>0$},
\end{equation*}
where $v(\cdot)$ is the dual value function in \eqref{eq:dvfabs}.

\end{lemma}

\begin{proof}

For $n\in\mathbb{N}$ denote by $\mathcal{B}_{n}$ the set of elements
in $L^{0}_{+}(\mu)$ lying in a ball of radius $n$:
\begin{equation*}
\mathcal{B}_{n} := \left\{h\in L^{0}_{+}(\mu): h\leq
  n,\,\mu-\mathrm{a.e.}\right\}. 
\end{equation*}
The sets $(\mathcal{B}_{n})_{n\in\mathbb{N}}$ are
$\sigma(L^{\infty},L^{1})$-compact. Because each $g\in\mathcal{C}(x)$
is $\mu$-integrable, $\mathcal{C}(x)$ is a closed, convex subset of
the vector space $L^{1}(\mu)$, so we apply the minimax theorem as
given in Theorem \ref{thm:rminimax} to the compact set
$\mathcal{B}_{n}$ ($n$ fixed) and the set $\mathcal{C}(x)$, with the
function $f(g,h):=\int_{\mathbf{\Omega}}(V(h)+gh)\ud\mu$, for
$g\in\mathcal{C}(x),\,h\in\mathcal{B}_{n}$, to give
\begin{equation}
\inf_{h\in\mathcal{B}_{n}}\sup_{g\in\mathcal{C}(x)}
\int_{\mathbf{\Omega}}(V(h)+gh)\ud\mu
= \sup_{g\in\mathcal{C}(x)}\inf_{h\in\mathcal{B}_{n}}
\int_{\mathbf{\Omega}}(V(h)+gh)\ud\mu.
\label{eq:mm}
\end{equation}
By the bipolarity relation $\mathcal{D}=\mathcal{C}^{\circ}$ in
\eqref{eq:bp2}, an element $h\in L^{0}_{+}(\mu)$ lies in
$\mathcal{D}(y)$ if and only if
$\sup_{g\in\mathcal{C}(x)}\int_{\mathbf{\Omega}}gh\ud\mu\leq
xy$. Thus, the limit as $n\to\infty$ on the left-hand-side of
\eqref{eq:mm} is given as
\begin{equation}
\lim_{n\to\infty}\inf_{h\in\mathcal{B}_{n}}\sup_{g\in\mathcal{C}(x)}
\int_{\mathbf{\Omega}}(V(h)+gh)\ud\mu =
\inf_{y>0}\inf_{h\in\mathcal{D}(y)}\left(\int_{\mathbf{\Omega}}V(h)\ud\mu
+ xy\right) = \inf_{y>0}[v(y)+xy].
\label{eq:lhsmm}
\end{equation}
Now consider the right-hand-side of \eqref{eq:mm}. Define
\begin{equation*}
U_{n}(x):= \inf_{0<y\leq n}[V(y)+xy], \quad x>0, \quad n\in\mathbb{N}.
\end{equation*}
The right-hand-side of \eqref{eq:mm} is then given as
\begin{equation*}
\sup_{g\in\mathcal{C}(x)}\inf_{h\in\mathcal{B}_{n}}
\int_{\mathbf{\Omega}}(V(h)+gh)\ud\mu =
\sup_{g\in\mathcal{C}(x)}\int_{\mathbf{\Omega}}U_{n}(g)\ud\mu =:
u_{n}(x),  
\end{equation*}
so that taking the limit as $n\to\infty$ and equating this with the
limit obtained in \eqref{eq:lhsmm}, we have
\begin{equation}
\lim_{n\to\infty}u_{n}(x) = \inf_{y>0}[v(y)+xy] \geq u(x),
\label{eq:limun}  
\end{equation}
with the inequality due to the weak duality bound in
\eqref{eq:weakd}. Consequently, we will be done if we can now show
that we also have
\begin{equation*}
\lim_{n\to\infty}u_{n}(x) \leq u(x).  
\end{equation*}

Evidently, $(u_{n}(x))_{n\in\mathbb{N}}$ is a decreasing sequence
satisfying the limiting inequality in \eqref{eq:limun}. Let
$(\tilde{g}^{n})_{n\in\mathbb{N}}$ be a maximising sequence in
$\mathcal{C}(x)$ for $\lim_{n\to\infty}u_{n}(x)$, so such that
\begin{equation*}
\lim_{n\to\infty}\int_{\mathbf{\Omega}}U_{n}(\tilde{g}^{n})\ud\mu =
\lim_{n\to\infty}u_{n}(x).
\end{equation*}
The compactness lemma for $\mathcal{C}$, Lemma \ref{lem:Ccompact},
implies the existence of a sequence $(g^{n})_{n\in\mathbb{N}}$ in
$\mathcal{C}(x)$, with
$g^{n}\in\conv(\tilde{g}^{n},\tilde{g}^{n+1},\ldots)$, which converges
$\mu$-a.e. to an element $g\in\mathcal{C}(x)$. Now, $U_{n}(x)=U(x)$
for $x\geq I(n)$, where $I(\cdot)=-V^{\prime}(\cdot)$ is the inverse
of $U^{\prime}(\cdot)$ (and $U_{n}(\cdot)\to U(\cdot)$ as
$n\to\infty$). So we deduce from Lemma \ref{lem:Uplusg} that
the sequence $(U^{+}_{n}(g^{n}))_{n\in\mathbb{N}}$ is uniformly
integrable, and hence that
\begin{equation}
\lim_{n\to\infty}\int_{\mathbf{\Omega}}U^{+}_{n}(g^{n})\ud\mu =
\int_{\mathbf{\Omega}}U^{+}(g)\ud\mu.  
\label{eq:Uplusg1}
\end{equation}
On the other hand, from Fatou's lemma, we have
\begin{equation}
\lim_{n\to\infty}\int_{\mathbf{\Omega}}U^{-}_{n}(g^{n})\ud\mu \geq
\int_{\mathbf{\Omega}}U^{-}(g)\ud\mu,
\label{eq:Uminusg1}
\end{equation}
so \eqref{eq:Uplusg1} and \eqref{eq:Uminusg1} give
\begin{equation}
\lim_{n\to\infty}\int_{\mathbf{\Omega}}U_{n}(g^{n})\ud\mu \leq
\int_{\mathbf{\Omega}}U(g)\ud\mu.
\label{eq:limUg}
\end{equation}
Finally, using concavity of $U_{n}(\cdot)$ and \eqref{eq:limUg}, we obtain
\begin{equation*}
\lim_{n\to\infty}u_{n}(x) =
\lim_{n\to\infty}\int_{\mathbf{\Omega}}U_{n}(\tilde{g}^{n})\ud\mu \leq
\lim_{n\to\infty}\int_{\mathbf{\Omega}}U_{n}(g^{n})\ud\mu \leq
\int_{\mathbf{\Omega}}U(g)\ud\mu \leq u(x),
\end{equation*}
and the proof is complete.

\end{proof}

We now move on to the dual side of the analysis. We begin with a
similar compactness lemma to Lemma \ref{lem:Ccompact}, but now for the
dual domain. The proof is identical to the proof of Lemma
\ref{lem:Ccompact} so is omitted.

\begin{lemma}[Compactness lemma for $\mathcal{D}$]
\label{lem:Dcompact}

Let $(\tilde{h}^{n})_{n\in\mathbb{N}}$ be a sequence in
$\mathcal{D}$. Then there exists a sequence $(h^{n})_{n\in\mathbb{N}}$
with $h^{n}\in\conv(\tilde{h}^{n}, \tilde{h}^{n+1},\ldots)$, which
converges $\mu$-a.e. to an element $h\in\mathcal{D}$ that is
$\mu$-a.e. finite.

\end{lemma}

Next, we have an analogous result to Lemma \ref{lem:Uplusg}, but for
the dual variables, concerning the uniform integrability of a sequence
$(V^{-}(h^{n}))_{n\in\mathbb{N}}$ for $h^{n}\in\mathcal{D}(y)$ (which
will subsequently lead to a lemma on existence and uniqueness of the
dual optimiser). The proof is in the style of Kramkov and
Schachermayer \cite[Lemma 1]{ks03}, but there the technique was
applied to a corresponding primal result akin to Lemma
\ref{lem:Uplusg}.

\begin{lemma}[Uniform integrability of
$(V^{-}(h^{n}))_{n\in\mathbb{N}},\,h^{n}\in\mathcal{D}(y)$] 
\label{lem:Vminush}

Let $(h^{n})_{n\in\mathbb{N}}$ be a sequence in $\mathcal{D}(y)$, for
any fixed $y>0$. The sequence $(V^{-}(h^{n}))_{n\in\mathbb{N}}$ is
uniformly integrable. 
  
\end{lemma}

\begin{proof}

Fix $y>0$. If $V(\infty)\geq 0$ there is nothing to prove, so assume
$V(\infty)<0$.

If the sequence $(V^{-}(h^{n}))_{n\in\mathbb{N}}$ is not uniformly
integrable, then, passing if need be to a subsequence still denoted by
$(h^{n})_{n\in\mathbb{N}}$, we can find a constant $\alpha>0$ and a
disjoint sequence $(A_{n})_{n\in\mathbb{N}}$ of sets of
$(\mathbf{\Omega},\mathcal{G})$ (so
$A_{n}\in\mathcal{G},\,n\in\mathbb{N}$ and $A_{i}\cap A_{j}=\emptyset$
if $i\neq j$) such that
\begin{equation*}
\int_{\mathbf{\Omega}}V^{-}(h^{n})\mathbbm{1}_{A_{n}}\ud\mu \geq
\alpha, \quad n\in\mathbb{N}.  
\end{equation*}
(See for example Pham \cite[Corollary A.1.1]{pham09}.) Define, for
some $h^{0}\in\mathcal{D}$, a sequence $(f^{n})_{n\in\mathbb{N}}$ of
elements in $L^{0}_{+}(\mu)$ by
\begin{equation*}
f^{n} := y_{0}h^{0} + \sum_{k=1}^{n}h^{k}\mathbbm{1}_{A_{k}},   
\end{equation*}
where $y_{0}:=\inf\{y>0:\,V(y)\leq 0\}$. For any $g\in\mathcal{C}$
we have
\begin{equation*}
\int_{\mathbf{\Omega}}gf^{n}\ud\mu =
\int_{\mathbf{\Omega}}g\left(y_{0}h^{0} +
\sum_{k=1}^{n}h^{k}\mathbbm{1}_{A_{k}}\right)\ud\mu \leq y_{0} +
\sum_{k=1}^{n}\int_{\mathbf{\Omega}}gh^{k}\mathbbm{1}_{A_{k}}\ud\mu
\leq y_{0} + ny.
\end{equation*}
Thus, $f^{n}\in\mathcal{D}(y_{0}+ny),\,n\in\mathbb{N}$.

On the other hand, since $V^{-}(\cdot)$ is non-negative and
non-decreasing,
\begin{eqnarray*}
\int_{\mathbf{\Omega}}V(f^{n})\ud\mu & = &
-\int_{\mathbf{\Omega}}V^{-}(f^{n})\ud\mu \\
& = & -\int_{\mathbf{\Omega}}V^{-}\left(y_{0} +
\sum_{k=1}^{n}h^{k}\mathbbm{1}_{A_{k}}\right)\ud\mu \\
& \leq & -\int_{\mathbf{\Omega}}
V^{-}\left(\sum_{k=1}^{n}h^{k}\mathbbm{1}_{A_{k}}\right)\ud\mu \\
& = & -\sum_{k=1}^{n}\int_{\mathbf{\Omega}}
V^{-}\left(h^{k}\mathbbm{1}_{A_{k}}\right)\ud\mu \leq -\alpha n.
\end{eqnarray*}
Therefore,
\begin{equation*}
\liminf_{z\to\infty}\frac{v(z)}{z} =
\liminf_{n\to\infty}\frac{v(y_{0}+ny)}{y_{0}+ny} \leq
\liminf_{n\to\infty}\frac{\int_{\mathbf{\Omega}}V(f^{n})\ud\mu}{y_{0}+ny}
\leq \liminf_{n\to\infty}\left(\frac{-\alpha n}{y_{0}+ny}\right) =
-\frac{\alpha}{y} < 0,
\end{equation*}
which contradicts the limiting weak duality bound in
\eqref{eq:limiting}. This contradiction establishes the result.

\end{proof}

One can can now proceed to prove existence of a unique optimiser in
the dual problem. The proof is on similar lines to the proof of primal
existence (Lemma \ref{lem:primexis}), with adjustments for
minimisation as opposed to maximisation, convexity of $V(\cdot)$
replacing concavity of $U(\cdot)$ and Lemma \ref{lem:Vminush}
replacing Lemma \ref{lem:Uplusg}. For brevity, therefore, the proof is
omitted.

\begin{lemma}[Dual existence]
\label{lem:dualexis}

The optimal solution $\widehat{h}(y)\in\mathcal{D}(y)$ to the dual
problem \eqref{eq:dvfabs} exists and is unique, so that $v(\cdot)$ is
strictly convex.  
  
\end{lemma}

We now move on to further characterise the derivatives of the value
functions, as well as the primal and dual optimisers. The first result
is on the derivative of the primal value value function $u(\cdot)$ at
infinity (equivalently, the derivative of the dual value function
$v(\cdot)$ at zero). Once again, because of the switching of the roles
of the primal and dual sets in our proofs compared with those of the
terminal wealth problem, the proof of the following lemma matches
closely the proof in Kramkov and Schachermayer \cite[Lemma 3.5]{ks99}
of the derivative of $v(\cdot)$ at infinity (giving the derivative of
$u(\cdot)$ at zero).

\begin{lemma}
\label{lem:uprime0}

The derivatives of the primal value function in \eqref{eq:vfabs} at
infinity and of the dual value function in \eqref{eq:dvfabs} at zero
are given by
\begin{equation}
u^{\prime}(\infty) := \lim_{x\to\infty}u^{\prime}(x) = 0, \quad
-v^{\prime}(0) := \lim_{y\downarrow 0}(-v^{\prime}(y)) = +\infty.   
\label{eq:uprime0}
\end{equation}

\end{lemma}

\begin{proof}

By the conjugacy result in Lemma \ref{lem:conjugacy} between the value
functions, the assertions in \eqref{eq:uprime0} are equivalent. We
shall prove the first assertion.

The primal value function $u(\cdot)$ is strictly concave and strictly
increasing, so there is a finite non-negative limit
$u^{\prime}(\infty):=\lim_{x\to\infty}u^{\prime}(x)$. Because
$U(\cdot)$ is increasing with $\lim_{x\to\infty}U^{\prime}(x)=0$, for
any $\epsilon>0$ there exists a number $C_{\epsilon}$ such that
$U(x)\leq C_{\epsilon}+\epsilon x,\,\forall\,x>0$. Using this, the
$L^{1}(\mu)$-boundedness of $\mathcal{C}$ (so that
$\int_{\mathbf{\Omega}}g\ud\mu\leq Kx,\,\forall\,g\in\mathcal{C}(x)$,
for some $K<\infty$) and l'H\^opital's rule, we have, with
$\int_{\mathbf{\Omega}}\ud\mu=:\delta>0$,
\begin{eqnarray*}
0 \leq \lim_{x\to\infty}u^{\prime}(x) =
\lim_{x\to\infty}\frac{u(x)}{x} & = &
\lim_{x\to\infty}\sup_{g\in\mathcal{C}(x)}\int_{\mathbf{\Omega}}
\frac{U(g)}{x}\ud\mu \\
& \leq & \lim_{x\to\infty}\sup_{g\in\mathcal{C}(x)}\int_{\mathbf{\Omega}}
\frac{C_{\epsilon}+\epsilon g}{x}\ud\mu \\
& \leq & \lim_{x\to\infty}\left(\frac{C_{\epsilon}\delta}{x} +
\epsilon K\right) = \epsilon K,
\end{eqnarray*}
and taking the limit as $\epsilon\downarrow 0$ gives the result.
  
\end{proof}

The final step in the series of lemmas that will furnish us with the
proof of Theorem \ref{thm:adt} is to characterise the derivative of
the primal value value function $u(\cdot)$ at zero (equivalently, the
derivative of the dual value function $v(\cdot)$ at infinity) along
with a duality characterisation of the primal and dual optimisers.

\begin{lemma}
\label{lem:derivatives}
  
\begin{enumerate}
  
\item The derivatives of the primal value function in \eqref{eq:vfabs}
at zero and of the dual value function in \eqref{eq:dvfabs} at
infinity are given by
\begin{equation}
u^{\prime}(0) := \lim_{x\downarrow 0}u^{\prime}(x) = +\infty, \quad
-v^{\prime}(\infty) := \lim_{y\to\infty}(-v^{\prime}(y)) = 0.   
\label{eq:mvprime0}
\end{equation}

\item For any fixed $x>0$, with $y=u^{\prime}(x)$ (equivalently
  $x=-v^{\prime}(y)$), the primal and dual optimisers
  $\widehat{g}(x),\widehat{h}(y)$ are related by
\begin{equation}
U^{\prime}(\widehat{g}(x)) = \widehat{h}(y) =
\widehat{h}(u^{\prime}(x)), \quad \mu\mbox{-a.e.},
\label{eq:primaldual}
\end{equation}
and satisfy
\begin{equation}
\int_{\mathbf{\Omega}}\widehat{g}(x)\widehat{h}(y)\ud\mu = xy =
xu^{\prime}(x).
\label{eq:saturation}
\end{equation}

\item The derivatives of the value functions satisfy the relations
\begin{equation}
xu^{\prime}(x) =
\int_{\mathbf{\Omega}}U^{\prime}(\widehat{g}(x))\widehat{g}(x)\ud\mu,
\quad yv^{\prime}(y) =
\int_{\mathbf{\Omega}}V^{\prime}(\widehat{h}(y))\widehat{h}(y)\ud\mu,
\quad x,y >0.
\label{eq:derivatives}
\end{equation}

\end{enumerate}

\end{lemma}

\begin{proof}

Recall the inequality \eqref{eq:VUbound}, which also applies to the
value functions because they are also conjugate by Lemma
\ref{lem:conjugacy}. We thus have, in addition to \eqref{eq:VUbound},
\begin{equation}
v(y) \geq u(x) -xy, \quad \forall\,x,y>0, \quad \mbox{with equality
iff $y=u^{\prime}(x)$}.
\label{eq:vubound}
\end{equation}
With $\widehat{g}(x)\in\mathcal{C}(x),\,x>0$ and
$\widehat{h}(y)\in\mathcal{D}(y),\,y>0$ denoting the primal and dual
optimisers, the bipolarity relations \eqref{eq:bp1} and \eqref{eq:bp2}
imply that we have
\begin{equation*}
\int_{\mathbf{\Omega}}\widehat{g}(x)\widehat{h}(y)\ud\mu \leq xy, \quad
x,y >0.  
\end{equation*}
Using this as well as \eqref{eq:VUbound} and \eqref{eq:vubound} we
have
\begin{equation}
 0 \leq \int_{\mathbf{\Omega}}\left(V(\widehat{h}(y)) -
U(\widehat{g}(x)) + \widehat{g}(x)\widehat{h}(y)\right)\ud\mu \leq
v(y) - u(x) + xy, \quad x,y>0,
\label{eq:fundineq}
\end{equation}
The right-hand-side of \eqref{eq:fundineq} is zero if and only if
$y=u^{\prime}(x)$, due to \eqref{eq:vubound}, and the non-negative
integrand must then be $\mu$-a.e. zero, which by \eqref{eq:VUbound}
can only happen if \eqref{eq:primaldual} holds, which establishes that
primal-dual relation.

Thus, for any fixed $x>0$ and with $y=u^{\prime}(x)$, and hence
equality in \eqref{eq:fundineq}, we have
\begin{eqnarray*}
0 & = & \int_{\mathbf{\Omega}}\left(V(\widehat{h}(y)) -
U(\widehat{g}(x)) + \widehat{g}(x)\widehat{h}(y)\right)\ud\mu \\
& = & v(y) - u(x) +
\int_{\mathbf{\Omega}}\widehat{g}(x)\widehat{h}(y)\ud\mu \\ 
& = & v(y) - u(x) + xy, \quad y=u^{\prime}(x),
\end{eqnarray*}
which implies that \eqref{eq:saturation} must hold. Inserting the
explicit form of $\widehat{h}(y)=U^{\prime}(\widehat{g}(x))$ into
\eqref{eq:saturation} yields the first relation in
\eqref{eq:derivatives}. Similarly, setting
$\widehat{g}(x)=I(\widehat{h}(y))=-V^{\prime}(\widehat{h}(y))$ into
\eqref{eq:saturation}, with $x=-v^{\prime}(y)$ (equivalent to
$y=u^{\prime}(x)$), yields the second relation in
\eqref{eq:derivatives}.

It remains to establish the relations in \eqref{eq:mvprime0}, which are
equivalent assertions. We shall prove the first one. This will use the
fact that $\mathcal{C}$ is a subset of $L^{1}(\mu)$. In the terminal
wealth case, one typically proves the second assertion using the
property that the dual domain lies within $L^{1}(\mathbb{P})$. This is
the switching of the roles of the primal and dual domains in the
consumption problem, that we have witnessed throughout this section.

From the first relation in \eqref{eq:derivatives} and the fact that
\begin{equation}
\int_{\mathbf{\Omega}}gh\ud\mu\leq xy, \quad
\forall\,g\in\mathcal{C}(x),h\in\mathcal{D}(y), \quad x,y>0,
\label{eq:absbc}
\end{equation}
we see that, for any $x>0$, we have
$U^{\prime}(\widehat{g}(x))\in\mathcal{D}(u^{\prime}(x))$. Thus, for
any $g\in\mathcal{C}$, \eqref{eq:absbc} implies that
\begin{equation}
u^{\prime}(x) \geq
\int_{\mathbf{\Omega}}U^{\prime}(\widehat{g}(x))g\ud\mu, \quad
\forall\,g\in\mathcal{C},
\label{eq:uprimeineq}
\end{equation}
which we shall make use of shortly.

Since $\mathcal{C}(x)$ is a subset of $L^{1}(\mu)$, we have
$\int_{\mathbf{\Omega}}\widehat{g}(x)\ud\mu\leq Kx$, for some
$K<\infty$, and hence
\begin{equation}
\int_{\mathbf{\Omega}}\frac{\widehat{g}(x)}{x}\ud\mu\leq K, \quad
\forall\,x>0.
\label{eq:whgxineq}
\end{equation}
Using Fatou's lemma in \eqref{eq:whgxineq} we have
\begin{equation*}
K \geq \liminf_{x\downarrow 0}
\int_{\mathbf{\Omega}}\frac{\widehat{g}(x)}{x}\ud\mu \geq
\int_{\mathbf{\Omega}}\liminf_{x\downarrow 0}
\left(\frac{\widehat{g}(x)}{x}\right)\ud\mu, 
\end{equation*}
which, given that $\widehat{g}(x)/x$ is non-negative, gives that
$\liminf_{x\downarrow 0}(\widehat{g}(x)/x)<\infty,\,\mu$-a.e.
Therefore, writing $\widehat{g}(x)=:x\widehat{g}^{x}$, which defines a
unique element $\widehat{g}^{x}\in\mathcal{C}$, we have
\begin{equation*}
\widehat{g}^{0} := \liminf_{x\downarrow 0}\widehat{g}^{x} = 
\liminf_{x\downarrow 0}\frac{\widehat{g}(x)}{x} < \infty, \quad
\mu\mbox{-a.e.}   
\end{equation*}
Using this property and applying Fatou's lemma to \eqref{eq:uprimeineq}
we obtain, on using $U^{\prime}(0)=+\infty$,
\begin{equation*}
+\infty \geq \liminf_{x\downarrow 0}u^{\prime}(x) \geq 
\liminf_{x\downarrow 0}
\int_{\mathbf{\Omega}}U^{\prime}(x\widehat{g}^{x})g\ud\mu \geq
\int_{\mathbf{\Omega}}\liminf_{x\downarrow 0}
U^{\prime}(x\widehat{g}^{x})g\ud\mu = +\infty, 
\end{equation*}
which gives us the first relation in \eqref{eq:mvprime0}.

\end{proof}

We have now established all results that give the duality in Theorem
\ref{thm:adt}, so let us confirm this.

\begin{proof}[Proof of Theorem \ref{thm:adt}]

Lemma \ref{lem:conjugacy} implies the relations \eqref{eq:conjugacy}
of item (i). The statements in item (ii) are implied by Lemma
\ref{lem:primexis} and Lemma \ref{lem:dualexis}. Items (iii) and (iv)
follow from Lemma \ref{lem:uprime0} and Lemma \ref{lem:derivatives}.

\end{proof}

We are almost ready to prove the concrete duality in Theorem
\ref{thm:consd}, because Theorem \ref{thm:adt} readily implies nearly
all of the assertions of Theorem \ref{thm:consd}. The outstanding
assertion is the characterisation of the optimal wealth process in
\eqref{eq:owp} and the associated uniformly integrable martingale
property of the deflated wealth plus cumulative deflated consumption
process
$\widehat{X}(x)\widehat{Y}(y) +
\int_{0}^{\cdot}\widehat{c}_{s}(x)\widehat{Y}_{s}(y)\ud s$. So we
proceed to establish these assertions in the proposition below, which
turns out to be interesting in its own right. We take as given the
other assertions of Theorem \ref{thm:consd}, and in particular the
optimal budget constraint in \eqref{eq:oihbc}. We shall confirm the
proof of Theorem \ref{thm:consd} in its entirety after the proof of
the next result.

\begin{proposition}[Optimal wealth process]
\label{prop:owp}

Given the saturated budget constraint equality in \eqref{eq:oihbc},
the optimal wealth process is characterised by \eqref{eq:owp}. The
process
\begin{equation*}
\widehat{M}_{t} := \widehat{X}_{t}(x)\widehat{Y}_{t}(y) +
\int_{0}^{t}\widehat{c}_{s}(x)\widehat{Y}_{s}(y)\ud s, \quad 0\leq
t <\infty, 
\end{equation*}
is a uniformly integrable martingale, converging to an integrable
random variable $\widehat{M}_{\infty}$, so the martingale extends to
$[0,\infty]$. The process $\widehat{X}(x)\widehat{Y}(y)$ is a
potential, that is, a non-negative supermartingale satisfying
$\lim_{t\to\infty}\mathbb{E}[\widehat{X}_{t}(x)\widehat{Y}_{t}(y)]=0$.
Moreover, $\widehat{X}_{\infty}(x)\widehat{Y}_{\infty}(y)=0$, almost
surely.

\end{proposition}

\begin{proof}

It simplifies notation if we take $x=y=1$, and is without loss of
generality: although $y=u^{\prime}(x)$ in \eqref{eq:oihbc}, one can
always multiply the utility function by an arbitrary constant so as
to ensure that $u^{\prime}(1)=1$.  We thus have the optimal budget
constraint
\begin{equation}
\mathbb{E}\left[\int_{0}^{\infty}\widehat{c}_{t}\widehat{Y}_{t}\ud
t\right] =1,
\label{eq:oihbc1}
\end{equation}
for $\widehat{c}\equiv\widehat{c}(1)\in\mathcal{A}$ and
$\widehat{Y}\equiv\widehat{Y}(1)\in\mathcal{Y}$. Since
$\widehat{c}\in\mathcal{A}$, we know there exists an optimal wealth
process $\widehat{X}\equiv\widehat{X}(1)$ and an associated optimal
trading strategy $\widehat{H}$, such that
\begin{equation*}
\widehat{X} = 1 + (\widehat{H}\cdot S) -
\int_{0}^{\cdot}\widehat{c}_{s}\ud s \geq 0,   
\end{equation*}
and such that
$\widehat{M}:=\widehat{X}\widehat{Y} +
\int_{0}^{\cdot}\widehat{c}_{s}\widehat{Y}_{s}\ud s$ is a
supermartingale over $[0,\infty)$. The supermartingale condition, by
the same arguments that led to the derivation of the budget constraint
in Lemma \ref{lem:ihbc}, leads to the inequality
$\mathbb{E}\left[\int_{0}^{\infty}\widehat{c}_{t}\widehat{Y}_{t}\ud
  t\right] \leq1$ instead of the equality
\eqref{eq:oihbc1}. Similarly, if the supermartingale is strict, we
get a strict inequality in place of \eqref{eq:oihbc1}. We thus deduce
that $\widehat{M}$ must be a martingale over $[0,\infty)$. We shall
show that this extends to $[0,\infty]$, along with the other claims in
the proposition.

Since $\widehat{M}$ is a martingale, the (non-negative c\`adl\`ag)
deflated wealth process $\widehat{X}\widehat{Y}$ is a martingale minus
a non-decreasing process, so is a non-negative c\`adl\`ag
supermartingale, and thus (by Cohen and Elliott \cite[Corollary
5.2.2]{ce15}, for example) converges to an integrable limiting random
variable
$\widehat{X}_{\infty}\widehat{Y}_{\infty}:=
\lim_{t\to\infty}\widehat{X}_{t}\widehat{Y}_{t}$ (and moreover
$\widehat{X}_{t}\widehat{Y}_{t}\geq
\mathbb{E}[\widehat{X}_{\infty}\widehat{Y}_{\infty}],\,t\geq 0$). The
non-decreasing integral in $\widehat{M}$ clearly also converges to an
integrable random variable, by virtue of the budget constraint. Thus,
$\widehat{M}$ also converges to an integrable random variable
$\widehat{M}_{\infty}:=\widehat{X}_{\infty}\widehat{Y}_{\infty}+
\int_{0}^{\infty}\widehat{c}_{t}\widehat{Y}_{t}\ud t$. By Protter
\cite[Theorem I.13]{protter}, the extended martingale over
$[0,\infty]$, $(\widehat{M}_{t})_{t\in[0,\infty]}$ is then uniformly
integrable, as claimed.

The martingale condition gives
\begin{equation*}
\mathbb{E}\left[\widehat{X}_{t}\widehat{Y}_{t} +
\int_{0}^{t}\widehat{c}_{s}\widehat{Y}_{s}\ud s\right] = 1, \quad 0\leq
t <\infty.
\end{equation*}
Taking the limit as $t\to\infty$, using monotone convergence
in the second term within the expectation and utilising
\eqref{eq:oihbc1} yields
\begin{equation*}
\lim_{t\to\infty}\mathbb{E}[\widehat{X}_{t}\widehat{Y}_{t}] = 0,
\end{equation*}
so that $\widehat{X}\widehat{Y}$ is a potential, as claimed.

Using the uniform integrability of $\widehat{M}$ and taking the limit as
$t\to\infty$ in $\mathbb{E}[\widehat{M}_{t}]=1,\,t\geq 0$, we have
\begin{equation*}
1 = \lim_{t\to\infty}\mathbb{E}[\widehat{M}_{t}] =
\mathbb{E}\left[\lim_{t\to\infty}\widehat{M}_{t}\right] =
\mathbb{E}[\widehat{X}_{\infty}\widehat{Y}_{\infty}] + 1,
\end{equation*}
on using \eqref{eq:oihbc1}. Hence, we get
$\mathbb{E}[\widehat{X}_{\infty}\widehat{Y}_{\infty}]=0$ and, since
$\widehat{X}_{\infty}\widehat{Y}_{\infty}$ is non-negative, we deduce
that $\widehat{X}_{\infty}\widehat{Y}_{\infty}=0$, almost surely as
claimed. 

We can now assemble these ingredients to arrive at the optimal wealth
process formula \eqref{eq:owp}. Applying the martingale condition
again, this time over $[t,u]$ for some $t\geq 0$, we have
\begin{equation*}
\mathbb{E}\left[\left.\widehat{X}_{u}\widehat{Y}_{u} +
\int_{0}^{u}\widehat{c}_{s}\widehat{Y}_{s}\ud
s\right\vert\mathcal{F}_{t}\right] = \widehat{X}_{t}\widehat{Y}_{t} +
\int_{0}^{t}\widehat{c}_{s}\widehat{Y}_{s}\ud
s, \quad 0\leq t\leq u <\infty.
\end{equation*}
Taking thew limit as $u\to\infty$ and using the uniform integrability
of $\widehat{M}$ we obtain
\begin{equation*}
\mathbb{E}\left[\left.\lim_{u\to\infty}\left(\widehat{X}_{u}\widehat{Y}_{u} +
\int_{0}^{u}\widehat{c}_{s}\widehat{Y}_{s}\ud
s\right)\right\vert\mathcal{F}_{t}\right] = \widehat{X}_{t}\widehat{Y}_{t} +
\int_{0}^{t}\widehat{c}_{s}\widehat{Y}_{s}\ud
s, \quad t\geq 0,
\end{equation*}
which, on using $\widehat{X}_{\infty}\widehat{Y}_{\infty}=0$,
re-arranges to
\begin{equation*}
\widehat{X}_{t}\widehat{Y}_{t} =
\mathbb{E}\left[\left.\int_{t}^{\infty}\widehat{c}_{s}\widehat{Y}_{s}\ud
s\right\vert\mathcal{F}_{t}\right], \quad t\geq 0,
\end{equation*}
which establishes \eqref{eq:owp}, and the proof is complete.

\end{proof}

\begin{proof}[Proof of Theorem \ref{thm:consd}]

Given the definitions of the sets $\mathcal{C}(x)$ and
$\mathcal{D}(y)$ in \eqref{eq:Cx} and \eqref{eq:Dy}, respectively, and
the identification of the abstract value functions in \eqref{eq:vfabs}
and \eqref{eq:dvfabs} with their concrete counterparts in
\eqref{eq:vf} and \eqref{eq:dvf}, Theorem \ref{thm:adt} implies all
the assertions of Theorem \ref{thm:consd}, with the exception of the
optimal wealth process formula \eqref{eq:owp} and the uniform
integrability of $\widehat{X}(x)\widehat{Y}(y) +
\int_{0}^{\cdot}\widehat{c}_{s}(x)\widehat{Y}_{s}(y)\ud s$, which are
established by Proposition \ref{prop:owp}.
  
\end{proof}

\section{An example: Bessel process with stochastic
volatility and correlation}
\label{sec:example}

We end with an example of an infinite horizon consumption problem in
an incomplete market model with strict local martingale deflators,
which is covered in our framework.

\begin{example}[Three-dimensional Bessel process with stochastic
  volatility and correlation, CRRA utility]
\label{examp:3dbessel}

Take an infinite horizon complete stochastic basis
$(\Omega,\mathcal{F},\mathbb{F}:=(\mathcal{F}_{t})_{t\geq
  0},\mathbb{P})$, with $\mathbb{F}$ satisfying the usual hypotheses.
Let $(W,W^{\perp})$ be a two-dimensional Brownian motion. We take
$\mathbb{F}$ to be the augmented filtration generated by
$(W,W^{\perp})$.

Let $B$ denote the process which solves the stochastic differential
equation
\begin{equation*}
\ud B_{t} = \frac{1}{B_{t}}\ud t + \ud W_{t} =: \lambda_{t}\ud t + \ud
W_{t}, \quad B_{0}=1.  
\end{equation*}
The process $B$ is the well-known three-dimensional Bessel
process. The process $\lambda:=1/B$ will be the market price of risk
of a stock with price process $S$ and stochastic volatility process
$Y>0$, driven by the correlated Brownian motion
$\widetilde{W}:=\rho W+\sqrt{1-\rho^{2}}W^{\perp}$, and with
$\rho\in[-1,1]$ some $\mathbb{F}$-adapted stochastic correlation. We
need not specify the dynamics of $Y$ or $\rho$ any further for the
purposes of the example. The stock price dynamics are given by
\begin{equation*}
\ud S_{t} = Y_{t}S_{t}\ud B_{t} = Y_{t}S_{t}(\lambda_{t}\ud t + \ud
W_{t}).  
\end{equation*}

Take a constant relative risk aversion (CRRA) utility function:
$U(x):=x^{p}/p,\,p<1,p\neq 0,\,x>0$. The results for logarithmic
utility $U(\cdot)=\log(\cdot)$ can be recovered by setting $p=0$ in
the final formulae, and this can be verified by carrying out the
analysis directly for that case. Take the measure $\kappa$ to be given
by $\ud\kappa_{t}=\e^{-\alpha t}\ud t$, for a positive discount rate
$\alpha$, so that $\gamma_{t}=\e^{\alpha t},\,t\geq 0$. The primal
value function is
\begin{equation*}
u(x) :=
\sup_{c\in\mathcal{A}(x)}\mathbb{E}\left[\int_{0}^{\infty}\e^{-\alpha
t}U(c_{t})\ud t\right], \quad x>0.
\end{equation*}
The wealth process incorporating consumption satisfies
\begin{equation*}
\ud X_{t} = Y_{t}\pi_{t}(\lambda_{t}\ud t + \ud W_{t}) - c_{t}\ud t,
\quad X_{0}=x, 
\end{equation*}
where $\pi=HS$ is the trading strategy expressed in terms of the
wealth placed in the stock, with $H$ the process for the number of
shares.

With $\mathcal{E}(\cdot)$ denoting the stochastic exponential, the
deflators in this model are given by local martingale deflators of the
form
\begin{equation}
Z:=\mathcal{E}(-\lambda\cdot W - \psi\cdot W^{\perp}),  
\label{eq:Zpsi}
\end{equation}
for an arbitrary process $\psi$ satisfying
$\int_{0}^{t}\psi^{2}_{s}\ud s<\infty$ almost surely for all
$t\geq 0$, with each such $\psi$ leading to a different deflator: this
market is of course incomplete. In the case that $Y$ and $\rho$ are
deterministic, the market is complete and there is a unique deflator
$Z^{(0)}:=\mathcal{E}(-\lambda\cdot W)$. It is well-known (see for
instance Larsen \cite[Example 2.2]{larsen09}) that $Z^{(0)}$ is a
strict local martingale and, what is more, that $Z^{(0)}=\lambda$ and
that $\lambda$ is square integrable. The strict local martingale
property is inherited by $Z$ in \eqref{eq:Zpsi}, for any choice of
integrand $\psi$.

The deflated wealth plus cumulative deflated consumption
process $M$ is then given by
\begin{equation}
M_{t} := X_{t}Z_{t} + \int_{0}^{t}c_{s}Z_{s}\ud s = x +
\int_{0}^{t}Z_{s}(Y_{s}\pi_{s}-\lambda_{s}X_{s})\ud W_{s} -
\int_{0}^{t}X_{s}Z_{s}\psi_{s}\ud W^{\perp}_{s}, \quad
t\geq 0,
\label{eq:M}
\end{equation}
which is a non-negative local martingale and thus a supermartingale.

The convex conjugate of the utility function is $V(y):=-y^{q}/q,\,y>0$,
where $q<1,\,q\neq 0$ is the conjugate variable to $p$, satisfying
$1-q=(1-p)^{-1}$. The dual value function is given by
\begin{equation*}
v(y) :=
\inf_{Z\in\mathcal{Z}}\mathbb{E}\left[\int_{0}^{\infty}\e^{-\alpha
t}V(yZ_{t}\e^{\alpha t})\ud t\right], \quad y>0.
\end{equation*}
Denote the unique dual minimiser by $\widehat{Z}$, given by
\begin{equation*}
\widehat{Z} :=\mathcal{E}(-\lambda\cdot W - \widehat{\psi}\cdot
W^{\perp}),
\end{equation*}
for some optimal integrand $\widehat{\psi}$ in \eqref{eq:Zpsi}. For
use below, define the non-negative martingale $H$ by
\begin{equation*}
H_{t} := \mathbb{E}\left[\left.
\int_{0}^{\infty}\e^{-\alpha(1-q)s}\widehat{Z}^{q}_{s}\ud
s\right\vert\mathcal{F}_{t}\right], \quad t\geq 0. 
\end{equation*}

Using Theorem \ref{thm:consd}, and in particular
\eqref{eq:pdc}, the optimal consumption process is given by
\begin{equation}
(\widehat{c}_{t}(x))^{-(1-p)} = u^{\prime}(x)\e^{\alpha
t}\widehat{Z}_{t}, \quad t\geq 0.
\label{eq:chat}
\end{equation}
By \eqref{eq:oihbc} the optimisers satisfy the saturated budget
constraint
\begin{equation}
\mathbb{E}\left[\int_{0}^{\infty}\widehat{c}_{t}(x)\widehat{Z}_{t}\ud
t\right]=x.
\label{eq:sbc}
\end{equation}
The relations \eqref{eq:chat} and \eqref{eq:sbc} yield
\begin{equation*}
\widehat{c}_{t}(x) = \frac{x}{H_{0}}\e^{-\alpha(1-q)
t}\widehat{Z}^{-(1-q)}_{t}, \quad t\geq 0.
\end{equation*}
Using \eqref{eq:owp}, the optimal wealth process is then given by
\begin{equation*}
\widehat{X}_{t}(x)\widehat{Z}_{t} = \frac{x}{H_{0}}\mathbb{E}\left[\left.
\int_{t}^{\infty}\e^{-\alpha(1-q)s}\widehat{Z}^{q}_{s}\ud
s\right\vert\mathcal{F}_{t}\right], \quad t\geq 0.  
\end{equation*}
More pertinently, the optimal martingale $\widehat{M}$, corresponding
to the process in \eqref{eq:M} at the optimum, is computed as
\begin{equation*}
\widehat{M}_{t} := \widehat{X}_{t}(x)\widehat{Z}_{t} +
\int_{0}^{t}\widehat{c}_{s}(x)\widehat{Z}_{s}\ud s =
\frac{x}{H_{0}}H_{t}, \quad t\geq 0,
\end{equation*}
so is indeed a martingale.

By martingale representation, $\widehat{M}$ will have a stochastic
integral representation which, without loss of generality, can be
written in the form
\begin{equation*}
\widehat{M}_{t} = x +
\int_{0}^{t}\widehat{Z}_{s}\widehat{X}_{s}(x)(\varphi_{s}
-q\lambda_{s})\ud W_{s} +
\int_{0}^{t}\widehat{Z}_{s}\widehat{X}_{s}(x)\beta_{s}\ud
W^{\perp}_{s}, \quad t\geq 0,  
\end{equation*}
for some integrands $\varphi,\beta$. Comparing with the representation
in \eqref{eq:M} at the optimum yields the optimal trading strategy in
terms of the optimal portfolio proportion
$\widehat{\theta}:=\widehat{\pi}/\widehat{X}(x)$ and the optimal
integrand $\widehat{\psi}$ in the form
\begin{equation*}
\widehat{\theta}_{t} := \frac{\widehat{\pi}_{t}}{\widehat{X}_{t}(x)} =
\frac{\lambda_{t}}{Y_{t}(1-p)} + \frac{\varphi_{t}}{Y_{t}}, \quad
\widehat{\psi}_{t} = -\beta_{t}, \quad t\geq 0.  
\end{equation*}
In particular, the process $\varphi$ records the correction to the
Merton-type strategy $\lambda/(Y(1-p))$ due to the stochastic
volatility and correlation.

This is as far as one can go without computing explicitly the dual
minimiser $\widehat{Z}$, which is typically impossible in closed form
for power utility. For the special case of logarithmic utility, one
can set $p=0$ and $q=0$ in the results for power utility, to show that
the process $H=1/\alpha$ is constant, and $\widehat{M}=x$ is also
constant, yielding
\begin{equation*}
\widehat{\theta}_{t} = \frac{\lambda_{t}}{Y_{t}}, \quad
\widehat{\psi}_{t} = 0, \quad t\geq 0,
\end{equation*}
giving the classic myopic trading strategy for logarithmic utility
and, in particular, that the dual optimiser is the minimal deflator:
$\widehat{Z}=Z^{(0)}=\mathcal{E}(-\lambda\cdot W)$. The optimal
consumption and wealth processes are given explicitly as
\begin{equation*}
\widehat{c}_{t}(x) = \alpha\e^{-\alpha t}\frac{x}{Z^{(0)}_{t}}, \quad
\widehat{X}_{t}(x) = \e^{-\alpha t}\frac{x}{Z^{(0)}_{t}}, \quad t\geq 0,
\end{equation*}
so that we have the classical relation
$\widehat{c}(x)=\alpha\widehat{X}(x)$, as is always the case for
infinite horizon logarithmic utility from consumption. The results for
logarithmic utility can of course be obtained by going directly
through the analysis from scratch in the manner above.

\end{example}

{\small
\bibliography{dfocunupbr_refs}

\begin{thebibliography}{10}

\bibitem{ak84}
{\sc J.-P. Aubin and I.~Ekeland}, {\em Applied nonlinear analysis}, Pure and
  Applied Mathematics (New York), John Wiley \& Sons, Inc., New York, 1984.
\newblock A Wiley-Interscience Publication.

\bibitem{bs99}
{\sc W.~Brannath and W.~Schachermayer}, {\em A bipolar theorem for
  {$L^0_+(\Omega,\mathcal{F},\mathbb P)$}}, in S\'eminaire de {P}robabilit\'es,
  {XXXIII}, vol.~1709 of Lecture Notes in Math., Springer, Berlin, 1999,
  pp.~349--354.

\bibitem{chauetal17}
{\sc H.~N. Chau, A.~Cosso, C.~Fontana, and O.~Mostovyi}, {\em Optimal
  investment with intermediate consumption under no unbounded profit with
  bounded risk}, J. Appl. Probab., 54 (2017), pp.~710--719.

\bibitem{ce15}
{\sc S.~N. Cohen and R.~J. Elliott}, {\em Stochastic calculus and
  applications}, Probability and its Applications, Springer, Cham, second~ed.,
  2015.

\bibitem{ck92}
{\sc J.~Cvitani\'{c} and I.~Karatzas}, {\em Convex duality in constrained
  portfolio optimization}, Ann. Appl. Probab., 2 (1992), pp.~767--818.

\bibitem{ds94}
{\sc F.~Delbaen and W.~Schachermayer}, {\em A general version of the
  fundamental theorem of asset pricing}, Math. Ann., 300 (1994), pp.~463--520.

\bibitem{ekq95}
{\sc N.~El~Karoui and M.-C. Quenez}, {\em Dynamic programming and pricing of
  contingent claims in an incomplete market}, SIAM J. Control Optim., 33
  (1995), pp.~29--66.

\bibitem{fkab98}
{\sc H.~F\"ollmer and Y.~M. Kabanov}, {\em Optional decomposition and
  {L}agrange multipliers}, Finance Stoch., 2 (1998), pp.~69--81.

\bibitem{fk97}
{\sc H.~F\"ollmer and D.~Kramkov}, {\em Optional decompositions under
  constraints}, Probab. Theory Related Fields, 109 (1997), pp.~1--25.

\bibitem{hp92}
{\sc C.-f. Huang and H.~Pag\`es}, {\em Optimal consumption and portfolio
  policies with an infinite horizon: existence and convergence}, Ann. Appl.
  Probab., 2 (1992), pp.~36--64.

\bibitem{kks16}
{\sc Y.~Kabanov, C.~Kardaras, and S.~Song}, {\em No arbitrage of the first kind
  and local martingale num\'eraires}, Finance Stoch., 20 (2016),
  pp.~1097--1108.

\bibitem{kk07}
{\sc I.~Karatzas and C.~Kardaras}, {\em The num\'eraire portfolio in
  semimartingale financial models}, Finance Stoch., 11 (2007), pp.~447--493.

\bibitem{kal86}
{\sc I.~Karatzas, J.~P. Lehoczky, S.~P. Sethi, and S.~E. Shreve}, {\em Explicit
  solution of a general consumption/investment problem}, Math. Oper. Res., 11
  (1986), pp.~261--294.

\bibitem{kls87}
{\sc I.~Karatzas, J.~P. Lehoczky, and S.~E. Shreve}, {\em Optimal portfolio and
  consumption decisions for a ``small investor'' on a finite horizon}, SIAM J.
  Control Optim., 25 (1987), pp.~1557--1586.

\bibitem{klsx91}
{\sc I.~Karatzas, J.~P. Lehoczky, S.~E. Shreve, and G.-L. Xu}, {\em Martingale
  and duality methods for utility maximization in an incomplete market}, SIAM
  J. Control Optim., 29 (1991), pp.~702--730.

\bibitem{ks98}
{\sc I.~Karatzas and S.~E. Shreve}, {\em Methods of mathematical finance},
  vol.~39 of Applications of Mathematics (New York), Springer-Verlag, New York,
  1998.

\bibitem{kz03}
{\sc I.~Karatzas and G.~\v{Z}itkovi\'{c}}, {\em Optimal consumption from
  investment and random endowment in incomplete semimartingale markets}, Ann.
  Probab., 31 (2003), pp.~1821--1858.

\bibitem{k12}
{\sc C.~Kardaras}, {\em Market viability via absence of arbitrage of the first
  kind}, Finance Stoch., 16 (2012), pp.~651--667.

\bibitem{ks99}
{\sc D.~Kramkov and W.~Schachermayer}, {\em The asymptotic elasticity of
  utility functions and optimal investment in incomplete markets}, Ann. Appl.
  Probab., 9 (1999), pp.~904--950.

\bibitem{ks03}
\leavevmode\vrule height 2pt depth -1.6pt width 23pt, {\em Necessary and
  sufficient conditions in the problem of optimal investment in incomplete
  markets}, Ann. Appl. Probab., 13 (2003), pp.~1504--1516.

\bibitem{k96}
{\sc D.~O. Kramkov}, {\em Optional decomposition of supermartingales and
  hedging contingent claims in incomplete security markets}, Probab. Theory
  Related Fields, 105 (1996), pp.~459--479.

\bibitem{larsen09}
{\sc K.~Larsen}, {\em Continuity of utility-maximization with respect to
  preferences}, Math. Finance, 19 (2009), pp.~237--250.

\bibitem{merton69}
{\sc R.~C. Merton}, {\em Lifetime portfolio selection under uncertainty: the
  continuous-time case}, Rev. Econ. Stat., 51 (1969), pp.~247--257.

\bibitem{most15}
{\sc O.~Mostovyi}, {\em Necessary and sufficient conditions in the problem of
  optimal investment with intermediate consumption}, Finance Stoch., 19 (2015),
  pp.~135--159.

\bibitem{pham09}
{\sc H.~Pham}, {\em Continuous-time stochastic control and optimization with
  financial applications}, vol.~61 of Stochastic Modelling and Applied
  Probability, Springer-Verlag, Berlin, 2009.

\bibitem{protter}
{\sc P.~E. Protter}, {\em Stochastic integration and differential equations},
  vol.~21 of Stochastic Modelling and Applied Probability, Springer-Verlag,
  Berlin, 2005.
\newblock Second edition. Version 2.1, Corrected third printing.

\bibitem{rogers}
{\sc L.~C.~G. Rogers}, {\em Duality in constrained optimal investment and
  consumption problems: a synthesis}, in Paris-{P}rinceton {L}ectures on
  {M}athematical {F}inance, 2002, vol.~1814 of Lecture Notes in Math.,
  Springer, Berlin, 2003, pp.~95--131.

\bibitem{rwvol1}
{\sc L.~C.~G. Rogers and D.~Williams}, {\em Diffusions, {M}arkov processes, and
  martingales. {V}ol. 1}, Cambridge Mathematical Library, Cambridge University
  Press, Cambridge, 2000.
\newblock Foundations, Reprint of the second (1994) edition.

\bibitem{sion58}
{\sc M.~Sion}, {\em On general minimax theorems}, Pacific J. Math., 8 (1958),
  pp.~171--176.

\bibitem{strasser85}
{\sc H.~Strasser}, {\em Mathematical theory of statistics}, vol.~7 of De
  Gruyter Studies in Mathematics, Walter de Gruyter \& Co., Berlin, 1985.
\newblock Statistical experiments and asymptotic decision theory.

\bibitem{sy98}
{\sc C.~Stricker and J.~A. Yan}, {\em Some remarks on the optional
  decomposition theorem}, in S\'eminaire de {P}robabilit\'es, {XXXII},
  vol.~1686 of Lecture Notes in Math., Springer, Berlin, 1998, pp.~56--66.

\bibitem{ts14}
{\sc K.~Takaoka and M.~Schweizer}, {\em A note on the condition of no unbounded
  profit with bounded risk}, Finance Stoch., 18 (2014), pp.~393--405.

\bibitem{z02}
{\sc G.~\v{Z}itkovi\'{c}}, {\em A filtered version of the bipolar theorem of
  {B}rannath and {S}chachermayer}, J. Theoret. Probab., 15 (2002), pp.~41--61.

\bibitem{sx92}
{\sc G.-L. Xu and S.~E. Shreve}, {\em A duality method for optimal consumption
  and investment under short-selling prohibition. {I}. {G}eneral market
  coefficients}, Ann. Appl. Probab., 2 (1992), pp.~87--112.

\end{thebibliography}
\bibliographystyle{siam}
}

\end{document}